\documentclass{llncs}

\usepackage{latexsym,graphicx}
\usepackage{amsmath,amssymb,enumerate,amsfonts}
\usepackage{authblk}
\usepackage{subfigure}
\usepackage{tikz}
\usetikzlibrary{arrows}
\usepackage[top=1in, bottom=1in, left=1in, right=1in]{geometry}
\usepackage[usenames,dvipsnames]{pstricks}
\usepackage{epsfig}
\usepackage{pst-grad} 
\usepackage{pstricks} 
\usepackage{graphicx} 

\newcommand{\OPT}{\textup{OPT}}
\newcommand{\job}{\mathcal{J}}
\newcommand{\machine}{\mathcal{M}}
\newcommand{\act}{\mathbb{A}}
\newcommand{\C}{\mathbb{C}}
\newcommand{\thres}{t + \lfloor \frac{W}{2} \rfloor}
\newcommand{\level}{\textsc{Level}}
\newcommand{\allorient}{\Psi}

\newcommand{\rock}{\mathbb{R}}
\newcommand{\pebble}{\mathbb{P}}
\newcommand{\father}{\mathcal{F}}
\newcommand{\child}{\mathcal{D}}

\newcommand{\inc}{\mathfrak{i}}
\newcommand{\out}{\mathfrak{o}}
\newcommand{\neu}{\mathfrak{n}}

\spnewtheorem{numberedClaim}{Claim}{\itshape}{\rmfamily}

\begin{document}

\title{A Combinatorial Approximation Algorithm for Graph Balancing with Light Hyper Edges}

\author{Chien-Chung Huang\inst{1} \and Sebastian Ott\inst{2}}
\institute{
	Chalmers University, G\"{o}teborg, Sweden \email{villars@gmail.com}
	\and Max-Planck-Institut f\"ur Informatik, Saarbr\"{u}cken, Germany \email{ott@mpi-inf.mpg.de} 
}

\maketitle

\begin{abstract}
Makespan minimization in restricted assignment $(R|p_{ij}\in \{p_j, \infty\}|C_{\max})$ 
is a classical problem in the field of machine scheduling. In a landmark paper in 1990~\cite{LST}, 
Lenstra, Shmoys, and Tardos gave a 2-approximation algorithm and proved that the problem 
cannot be approximated within 1.5 unless P=NP. The upper and lower bounds of the problem have been 
essentially unimproved in the intervening 25 years, despite several remarkable successful attempts 
in some special cases of the problem~\cite{ckl2015,Sgall,Svensson2012} recently. 

In this paper, we consider a special case called \emph{graph-balancing with light hyper edges}, where heavy jobs 
can be assigned to at most two machines while light jobs can be assigned to any number of machines. 
For this case, we present algorithms with approximation ratios strictly better than 2. Specifically, 

\begin{itemize}

\item \textbf{Two job sizes}: Suppose that light jobs have weight $w$ and heavy jobs have weight $W$, and $w < W$. 
We give a $1.5$-approximation algorithm (note that the current 1.5 lower bound is established in an even more restrictive setting~\cite{Asahiro11,SgallSODA2008}). Indeed, depending on the specific values of $w$ and $W$, 
sometimes our algorithm guarantees sub-1.5 approximation ratios. 

\item \textbf{Arbitrary job sizes}: Suppose that $W$ is the largest given weight, heavy jobs have weights 
in the range of $(\beta W, W]$, where $4/7\leq \beta < 1$, and light jobs have weights in the range of $(0,\beta W]$. 
We present a $(5/3+\beta/3)$-approximation algorithm. 
\end{itemize}

Our algorithms are purely combinatorial, without the need of solving a linear program 
as required in most other known approaches. 
 
\end{abstract}

\section{Introduction} 

Let $\job$ be a set of $n$ jobs and $\machine$ a set of $m$ machines. Each job $j \in \job$ has a \emph{weight} 
$w_j$ and can be assigned to a specific subset of the machines. An assignment $\sigma: \job \rightarrow \machine$ 
is a mapping where each job is mapped to a machine to which it can be assigned. The objective is to minimize 
the \emph{makespan}, defined as $\max_{i \in \machine} \sum_{j: \sigma(j)=i} w_j$. This is the classical 
{\sc makespan minimization in restricted assignment} $(R|p_{ij}\in \{p_j, \infty\}|C_{\max})$, itself a special case of the 
{\sc makespan minimization in unrelated machines} $(R||C_{\max})$, where a job $j$ has possibly different weight $w_{ij}$ 
on different machines $i \in \machine$. In the following, we just call them \textsc{restricted assignment} and \textsc{unrelated 
machine problem} for short. 

The first constant approximation algorithm for both problems is given by Lenstra, Shmoys, and Tardos~\cite{LST} 
in 1990, where the ratio is 2. They also show that \textsc{restricted assignment} (hence also the \textsc{unrelated 
machine problem}) cannot be approximated within 1.5 unless P=NP, even if there are only two job weights. The upper bound 
of 2 and the lower bound of 1.5 have been essentially unimproved in the intervening 25 years. How to close the gap 
continues to be 
one of the central topics in approximation algorithms. The recent book of Williamson and Shmoys~\cite{shmoyWillamsonBook} lists 
this as one of the ten open problems.

\subsubsection*{Our Result} We consider a special case of \textsc{restricted assignment}, called 
{\sc graph balancing with light hyper edges}, which is a generalization of the {\sc graph balancing} problem 
introduced by Ebenlendr, Kr\v{c}\'{a}l  and Sgall~\cite{SgallSODA2008}. There the restriction is that every job can be assigned to only two machines, and hence the problem can be interpreted in a graph-theoretic way: each machine is represented by a node, and each job is represented by an edge. The goal is to find an orientation of the edges so that the maximum weight sum of the edges oriented towards a node is minimized.
In our problem, jobs are partitioned into \emph{heavy} and \emph{light}, and we assume that heavy jobs can go to only two machines while light jobs can go to any number of machines\footnote{If some jobs can be assigned to just one machine, then it is the same as saying 
a machine has some \emph{dedicated load}. All our algorithms can handle arbitrary dedicated loads on the machines.}. 
In the graph-theoretic interpretation, light jobs are represented by hyper edges, while heavy jobs 
are represented by regular edges.

We present approximation algorithms with performance 
guarantee \emph{strictly better than} 2 in the following settings. For simplicity of presentation, 
we assume that all job weights $w_j$ are integral (this assumption is just for ease of exposition and can be easily 
removed). \\

\noindent \textbf{Two job sizes}: Suppose that heavy jobs are of weight $W$ and light jobs are of weight $w$, and 
$w < W$. We give a $1.5$-approximation algorithm, matching the general lower bound of \textsc{restricted assignment} (it should be noted that this lower bound is established in an even more restrictive setting~\cite{Asahiro11,SgallSODA2008}, where 
all jobs can only go to two machines and there are only two different job weights). This is the first time the lower bound 
is matched in a nontrivial case of \textsc{restricted assignment} (without specific restrictions on the job weight values). In fact, sometimes our algorithm achieves an approximation 
ratio strictly better than 1.5. Supposing that $w \leq \frac{W}{2}$, the ratio we get is $1 + \frac{\lfloor W/2 \rfloor}{W}$. \\[-5pt]

\noindent \textbf{Arbitrary job sizes}: Suppose that $\beta \in [4/7, 1)$ and $W$ is the largest given weight. 
A heavy job has weight in $(\beta W, W]$ while a light job has weight in $(0,\beta W]$. 
We give a $(5/3+\beta/3)$-approximation algorithm. \\[-5pt] 

\noindent Both algorithms have the running time of $\mathcal{O}\big(n^2 m^3 \log{(\sum_{j \in \job}w_j)}\big)$.\footnote{For simplicity, here we upper bound $\sum_{j \in \job}a_j$, where $a_j$ is the number of the machines $j$ can be assigned to, by $n m$.}\\

The general message of our result 
is clear: as long as the heaviest jobs have only two choices, it is relatively easy to break the barrier of 2 in the upper bound 
of \textsc{restricted assignment}. This should coincide with our intuition. The heavy jobs are in a sense 
the ``trouble-makers''. A mistake on them causes bigger damage than a mistake on lighter jobs. Restricting 
the choices of the heavy jobs thus simplifies the task.  \\

The original \textsc{graph balancing} problem assumes that all jobs can be assigned to only two machines and the algorithm of Ebenlendr et al.~\cite{SgallSODA2008}
gives a 1.75-approximation. According to~\cite{sgallCommunication}, their algorithm can be extended to our setting:
given any $\beta \in [0.5, 1)$, they can obtain a $(3/2 + \beta/2)$-approximation. 
Although this ratio is superior to ours, let us emphasize two interesting aspects of our approach. 

(1) The algorithm of Ebenlendr et al. requires solving a linear program (in fact, almost all known algorithms 
for the problem are LP-based), while our algorithms are purely combinatorial. In addition to 
the advantage of faster running time, our approach introduces new proof techniques (which 
do not involve linear programming duality).  

(2) In \textsc{graph balancing}, Ebenlendr et al. showed that with only two job weights and 
dedicated loads on the machines, their strongest LP has the integrality gap of 1.75, while we can 
break the gap.
Our approach thus offers a possible angle to circumvent the barrier posed by the integrality gap, and has the potential of seeing further improvement. 

Before explaining our technique in more detail, we should point out another interesting connection with 
a result of Svensson~\cite{Svensson2012} for general \textsc{restricted assignment}. 
He gave two local search algorithms, which terminate (but 
it is unknown whether in polynomial time) and (1) with two job weights $\{\epsilon,1\}$, $0< \epsilon < 1$, 
the returned solution has an approximation ratio of $5/3+ \epsilon$, and (2) with arbitrary job weights, the returned 
solution has an approximation ratio of $\approx 1.94$. It is worth noting that his analysis is done via the 
primal-duality of the configuration-LP (thus integrality gaps smaller than two for the configuration-LP are implied). 
With two job weights, our algorithm 
has some striking similarity to his algorithm. We are able to prove our algorithm terminates 
in polynomial time---but our setting is more restrictive. A very interesting direction 
for future work is to investigate how the ideas in the two algorithms can be related and combined. 

\subsubsection*{Our Technique}

Our approach is inspired by that of Gairing et al.~\cite{gairing04} for general \textsc{restricted assignment}. 
So let us first review their ideas. Suppose that a certain optimal makespan $t$ is guessed. Their core algorithm 
either (1) correctly reports that $t$ is an underestimate of $\OPT$, or (2) returns an assignment with makespan at most 
$t+W-1$. By a binary search on the smallest $t$ for which an assignment with makespan $t+W-1$ is returned, and the simple 
fact that $\OPT \geq W$, they guarantee the approximation ratio of 
$\frac{t+W-1}{\OPT} \leq 1 + \frac{W-1}{\OPT} \leq 2 -\frac{1}{W}$ (the first inequality holds because $t$ is the smallest 
number an assignment is returned by the core algorithm). Their core algorithm is a preflow-push algorithm. Initially all jobs 
are arbitrarily assigned. Their algorithm tries to redistribute the jobs from 
overloaded machines, i.e., those with load more than $t+W-1$, to those that are not. The redistribution 
is done by pushing the jobs around while updating the height labels (as commonly done in preflow-push algorithms). 
The critical thing is that after a polynomial number of steps, if there are still some overloaded machines, they 
use the height labels to argue that $t$ is a wrong guess, i.e., $\OPT \geq t+1$. 
Our contribution is a refined core algorithm in the same framework. With a guess $t$ of the optimal makespan, our core algorithm 
either (1) correctly reports that $\OPT \geq t+1$, or (2) returns an assignment with makespan at most 
$(5/3+\beta/3)t$. 

We divide all jobs into two categories, the 
\emph{rock jobs} $\rock$, and the \emph{pebble jobs} $\pebble$ (not to be confused with heavy and light jobs). 
The former consists of those with weights in $(\beta t, t]$ while the latter 
includes all the rest. We use the rock jobs to form a graph $G_{\rock}=(V, \rock)$, and assign the pebbles arbitrarily to the nodes. Our core algorithm will push around the pebbles so as to redistribute them. Observe that as $t \geq W$, all 
rocks are heavy jobs. So the formed graph $G_{\rock}$ has only simple edges (no hyper edges). As $\beta \geq 4/7$, 
if $\OPT \leq t$, then every node can receive at most one rock job in the optimal solution. 
In fact, it is easy to see that we can simply assume that the formed graph $G_{\rock}$ is a disjoint set of trees and cycles. Our entire task 
boils down to the following:

\begin{quote} Redistribute the pebbles so that there exists an orientation of the edges in $G_{\rock}$ in which each node has total load (from both 
rocks and pebbles) at most $(5/3+\beta/3)t$; and if not possible, gather evidence that $t$ is an underestimate. 
\end{quote}

Intuitively speaking, our algorithm maintains a certain \emph{activated set} $\act$ of nodes. Initially, this set includes those nodes 
whose total loads of pebbles cause conflicts in the orientation of the edges in $G_{\rock}$. A node ``reachable'' from a node in the 
activated set is also included into the set. (Node $u$ is reachable from node $v$ if a pebble in $v$ can be assigned to $u$.) 
Our goal is to push the pebbles among nodes in $\act$, so as to remove all conflicts in the edge 
orientation. Either we are successful in doing so, or we argue that the total load of all pebbles currently owned 
by the activated set, together with the total load of the rock jobs assigned to $\act$ in any \emph{feasible orientation} of the edges in $G_{\rock}$ (an orientation in $G_{\rock}$ is \emph{feasible} if every node receives at most one rock), 
is strictly larger than $t\cdot |\act|.$ The progress of our algorithm (hence its running time) is monitored by 
a potential function, which we show to be monotonically decreasing. 

The most sophisticated part of our algorithm is the ``activation strategy''. We initially add nodes into $\act$ if they 
cause conflicts in the orientation or can be (transitively) reached from such. However, sometimes we also include nodes 
that do not fall into the two categories. This is purposely done for two reasons: pushing pebbles from these nodes 
may help alleviate the conflict in edge orientation indirectly; and their presence in $\act$ strengthens the contradiction 
proof. 

Due to the intricacy of our main algorithm, 
we first present the algorithm for the two job weights case in Section~\ref{sec:twoWeights} and then present the main algorithm 
for the arbitrary weights in Section~\ref{sec:generalWeights}. The former algorithm is significantly simpler (with a straightforward 
activation strategy) and contains many ingredients of the ideas behind the main algorithm. 

\subsubsection*{Related Work}

For \textsc{restricted assignment}, besides the several recent advances mentioned earlier, see the survey of Leung and Li 
for other special cases~\cite{Leung2008}. For two job weights, Chakrabarti, Khanna and Li~\cite{ckl2015} showed that using 
the configuration-LP, they can obtain a $(2-\delta)$-approximation for a fixed $\delta>0$ (and note that there is no restriction 
on the number of machines a job can go to). 
Kolliopoulos and Moysoglou~\cite{km13} also considered the two job weights case. In the
\textsc{graph balancing} setting (with two job weights), they gave a 1.652-approximation algorithm using a flow technique 
(thus they also break the integrality gap in~\cite{Sgall}). 
They also show that the configuration-LP for \textsc{restricted assignment} with two job weights has an integrality gap of at most 1.883 (and this is further improved 
to 1.833 in~\cite{ckl2015}). 

For \textsc{unrelated machines}, Shchepin and Vakhania~\cite{shchepin2005} improved the approximation ratio to $2-1/m$. 
A combinatorial 2-approximation algorithm was given by Gairing, Monien, and Woclaw~\cite{gairingTCS}. 
Verschae and Wiese~\cite{Verschae2014} showed that the configuration-LP has integrality gap of 2, even if every job 
can be assigned to only two machines. They also showed that it is possible to achieve approximation ratios strictly better than 2 
if the job weights $w_{ij}$ respect some constraints.

\section{Preliminary}
\label{sec:pre}

Let $t$ be a guess of $\OPT$. Given $t$, our two core algorithms either report that $\OPT \geq t+1$, or return an assignment 
with makespan at most $1.5t$ or $(5/3+\beta/3)t$, respectively. We 
conduct a binary search on the smallest $t \in [W, \sum_{j \in \job}w_j]$ for which an assignment is returned 
by the core algorithms. This particular assignment is then the desired solution. 

We now explain the initial setup of the core algorithms. In our discussion, we will not distinguish 
a machine and a node. Let $dl(v)$ be the dedicated load of $v$, i.e., the sum of the weights of jobs 
that can only be assigned to $v$. We can assume that $dl(v) \leq t$ for all nodes $v$. 
Let $\job' \subseteq \job$ be the jobs that can be assigned to at least two machines. We divide 
$\job'$ into rocks $\rock$ and pebbles $\pebble$. A job $j \in \job'$ is a rock, 

\begin{itemize}

\item in the 2 job weights case (Section~\ref{sec:twoWeights}), if $w_j > t/2$ and $w_j = W$;

\item in the general job weights case (Section~\ref{sec:generalWeights}), if $w_j > \beta t$. 

\end{itemize}

A job $j \in \job'$ that is not a rock is a pebble. Define the graph $G_{\rock}=(V, \rock)$ as a graph with 
machines $\machine$ as node set and rocks $\rock$ as edge set. By our definition, a rock can be assigned 
to exactly two machines. So $G_{\rock}$ has only simple edges (no hyper edges). For the sake of convenience, we call the rocks just ``edges'', avoiding ambiguity by exclusively using the term ``pebble'' for the pebbles.

Suppose that $\OPT \leq t$. Then a machine can receive at most one rock in the optimal solution. If 
any connected component in $G_{\rock}$ has more than one cycle, we can immediately declare that 
$\OPT \geq t+1$. If a connected component in $G_{\rock}$ has exactly one cycle, we can direct all edges away from the cycle
and remove these edges, 
i.e., assign the rock to the node $v$ to which it is directed. W.L.O.G, we can assume that this rock 
is part of $v$'s dedicated load. (Also observe that then node $v$ must become an isolated node). 
Finally, we can eliminate cycles of length 2 in $G_{\rock}$ with the following simple reduction. If a pair of nodes $u$ and $v$ is connected by two distinct rocks $r1$ and $r2$, remove the two rocks, add $\min(w_{r1},w_{r2})$ to both $u$'s and $v$'s dedicated load, and introduce a new pebble of weight $|w_{r1}-w_{r2}|$ between $u$ and $v$. 
Let $\allorient$ denote the set of orientations in $G_{\rock}$ where each node has at most one incoming edge. 
We use a proposition to summarize the above discussion. 

\begin{proposition}\label{prop:preprocess} We can assume that
\begin{itemize}

\item the rocks in $\rock$ correspond to the edge set of the graph $G_{\rock}$, and all pebbles can be assigned 
to at least two machines;

\item the graph $G_{\rock}$ consists of disjoint trees, cycles (of length more than 2), and isolated nodes;

\item for each node $v \in V$, $dl(v) \leq t$;

\item if $\OPT \leq t$, then the orientation of the edges in $G_{\rock}$ in the optimal assignment must be one of those in $\allorient$. 

\end{itemize}

\end{proposition}

\section{The 2-Valued Case}
\label{sec:twoWeights}

In this section, we describe the core algorithm for the two job weights case, with the guessed makespan $t\geq W$. 
Observe that when $t\in[W,2w)$, if $\OPT\leq t$, then every node can receive at most one job (pebble or rock) in the optimal assignment. 
Hence, we can solve the problem exactly using the standard max-flow technique. 
So in the following, assume that $t \geq 2w$. 
Furthermore, let us first assume that $t < 2W$ (the case of $t\geq 2W$ will be discussed at the end of the section). 
Then the rocks have weight $W$ and the pebbles have weight $w$. 
Initially, the pebbles are arbitrarily assigned to the nodes. Let $pl(v)$ be the total weight of the pebbles assigned to node $v$. 

\begin{definition} A node $v$ is 

\begin{itemize}
\item \emph{uncritical}, if $dl(v) + pl(v) \leq 1.5t -W - w$;
\item \emph{critical}, if $dl(v) + pl(v) > 1.5t -W$;
\item \emph{hypercritical}, if $dl(v) + pl(v) > 1.5t$.
\end{itemize}
\end{definition}

(Notice that it is possible that a node is neither uncritical nor critical.) 

\begin{definition} Each tree, cycle, or isolated node in $G_{\rock}$ is a \emph{system}. A system is \emph{bad} if any of the following conditions holds. 

\begin{itemize}
\item It is a tree and has at least two critical nodes, or 
\item It is a cycle and has at least one critical node, or
\item It contains a hypercritical node. 
\end{itemize}

A system that is not bad is \emph{good}. 
\end{definition}

If all systems are good, then orienting the edges in each system such that every node has at most one incoming edge gives us a solution with makespan at most $1.5t$. So let assume that there is at least one bad system.

We next define the \emph{activated set} $\act$ of nodes constructively. Roughly speaking, we will move pebbles around 
the nodes in $\act$ so that either there is no more bad system left, or we argue that, in every feasible assignment, 
\emph{some} nodes in $\act$ cannot handle their total loads, thereby arriving at a contradiction. 

In the following, if a pebble in $u$ can be assigned to node $v$, we say $v$ is reachable from $u$. Node $v$ is reachable 
from $\act$ if $v$ is reachable from any node $u \in \act$. A node added into $\act$ is \emph{activated}. 

Informally, all nodes that cause a system to be bad are activated. 
A node reachable from $\act$ is also activated. Furthermore, suppose that a system is good and it has 
a critical node $v$ (thus the system cannot be a cycle). If any other node $u$ in the same system is activated, then 
so is $v$. We now give the formal procedure \textsc{Explore1} in Figure~\ref{fig:ExploreOne}. 
Notice that in the process of activating the 
nodes, we also define their \emph{levels}, which will be used later for the algorithm and the potential function.

\begin{figure}[!htb]
\centering
\fbox{ \begin{minipage}{0.85\linewidth}
\textsc{Explore1}\\[0.5em]
\textbf{Initialize} $\act := \{v| \mbox{$v$ is hypercritical, or $v$ is critical in a bad system}\}$. \\
\mbox{\hspace*{1em}} Set $\level(v):=0$ for all nodes in $\act$; $i:=0$.\\[0.6em]
\textbf{While} $\exists v \not \in \act$ reachable from $\act$ \textbf{do:}\\[0.2em]
\mbox{\hspace*{1em}} $i:= i+1$. \\
\mbox{\hspace*{1em}} $\act_i:= \{v \not \in \act| v \mbox{ reachable from $\act$}$\}. \\
\mbox{\hspace*{1em}} $\act_i':= \{v \not \in \act| \mbox{ $v$ is critical in a good system and $\exists u \in \act_i$ in the same system}\}$. \\
\mbox{\hspace*{1em}} Set $\level(v):=i$ for all nodes in $\act_i$ and $\act_i'$. \\
\mbox{\hspace*{1em}} $\act := \act \cup \act_i \cup \act_i'$. \\
For each node $v \not \in \act$, set $\level(v) = \infty$. 
\end{minipage} }
\caption{The procedure \textsc{Explore1}.}
\label{fig:ExploreOne}
\end{figure}

The next proposition follows straightforwardly from \textsc{Explore1}. 

\begin{proposition} The following holds. 
\begin{enumerate} 
\item All nodes reachable from $\act$ are in $\act$. 
\item Suppose that $v$ is reachable from $u \in \act$. Then $\level(v) \leq \level(u)+1$. 
\item If a node $v$ is critical and there exists another node $v' \in \act$ in the same system, then 
$\level(v) \leq \level(v')$. 

\item Suppose that node $v \in \act$ has $\level(v)=i > 0$. Then there exists another node $u \in \act$ 
with $\level(u)=i-1$ so that either $v$ is reachable from $u$, or there exists another node $v' \in \act$ reachable from $u$ 
with $\level(v')=i$ in the same system as $v$ and $v$ is critical. 
\end{enumerate}

\label{pro:first}
\end{proposition}

After \textsc{Explore1}, we apply the \textsc{Push} operation (if possible), defined as follows. 

\begin{definition} \textsc{Push} operation: push a pebble from $u^*$ to $v^*$ if the following conditions hold. 
\label{def:simplePush}
\begin{enumerate}

\item The pebble is at $u^*$ and it can be assigned to $v^*$. 
\item $\level(v^*)=\level(u^*)+1$. 
\item $v^*$ is uncritical, or $v^*$ is in a good system that remains good with an additional weight of $w$ at $v^*$.
\item Subject to the above three conditions, choose a node $u^*$ so that $\level(u^*)$ is minimized 
(if there are multiple candidates, pick any). 

\end{enumerate}

\end{definition}

Our algorithm can be simply described as follows. 

\begin{quote} \textbf{Algorithm~1}: As long as there is a bad system, apply \textsc{Explore1} and \textsc{Push} operation repeatedly. When there is no bad system left, return a solution with makespan at most $1.5t$.
If at some point, \textsc{push} is no longer possible, declare that $\OPT \geq t+1$.
\end{quote}

\begin{lemma} When there is at least one bad system and the \textsc{Push} operation is no longer possible, $\OPT \geq t+1$. 
\label{lem:firstContradiction}
\end{lemma}

\begin{proof} Let $\act(S)$ denote the set of activated nodes in system $S$. Recall that $\allorient$ denotes 
the set of all orientations in $G_{\rock}$ in which each node has at most one incoming edge. We prove 
the lemma via the following claim. 

\begin{numberedClaim} Let $S$ be a system. 
\label{cla:firstMain}
\begin{itemize}
\item Suppose that $S$ is bad. Then 
\begin{equation} W\cdot(\min_{\psi \in \allorient} \mbox{number of rocks to $\act(S)$ according to $\psi$}) + \sum_{v \in \act(S)} pl(v) + dl(v) 
> |\act(S)|t.
\label{equ:2valTooMuch}
\end{equation}
\item Suppose that $S$ is good. Then 
\begin{equation} W\cdot(\min_{\psi \in \allorient} \mbox{number of rocks to $\act(S)$ according to $\psi$}) + \sum_{v \in \act(S)} pl(v) + dl(v) 
> |\act(S)|t -w.
\label{equ:2valAlmostTooMuch}
\end{equation}
\end{itemize}
\end{numberedClaim}

Observe that the term $|\act(S)|t$ is the maximum total weight that all nodes in $\act(S)$ can handle if $\OPT \leq t$. 
As pebbles owned by nodes in $\act$ can only be assigned to the nodes 
in $\act$, by the pigeonhole principle, in all orientations $\psi \in \allorient$, and all possible assignments 
of the pebbles, at least one bad system $S$ has at least the same number of pebbles in $\act(S)$ as the current 
assignment, or a good system $S$ has at least one more pebble 
than it currently has in $\act(S)$. In both cases, we reach 
a contradiction. 

\qed\end{proof}

\noindent\emph{Proof of Claim~\ref{cla:firstMain}:}
First observe that in all orientations in $\allorient$, the nodes 
in $\act(S)$ have to receive at least $|\act(S)|-1$ rocks. If $S$ is a cycle, then 
the nodes in $\act(S)$ have to receive exactly $|\act(S)|$ rocks. 

Next observe that none of the nodes in $\act(S)$ is uncritical, since otherwise, by Proposition~\ref{pro:first}.4 and 
Definition~\ref{def:simplePush}.3, the \textsc{Push} 
operation would still be possible. By the same reasoning, if $S$ is a tree and $\act(S) \neq \emptyset$, at least one node $v\in \act(S)$ is critical; furthermore, if $|\act(S)| = 1$, this node $v$ satisfies $dl(v) + pl(v) > 1.5t-w$, as an additional weight of $w$ would make $v$ hypercritical. Similarly, if $S$ is an isolated node $v \in \act$, then $dl(v) + pl(v) > 1.5t-w$. 

We now prove the claim by the following case analysis. 

\begin{enumerate}

\item Suppose that $S$ is a good system and $\act(S) \neq \emptyset$. Then either $S$ is a tree and $\act(S)$ contains exactly one critical (but not hypercritical) node, or $S$ is an isolated node, or $S$ is a cycle and has no critical node. In the first case, if $|\act(S)| \geq 2$, 
the LHS of (\ref{equ:2valAlmostTooMuch}) is at least
\begin{eqnarray*}
(1.5t - W+1) + (|\act(S)|-1)(1.5t -W-w+1) + (|\act(S)|-1)W = \\
|\act(S)|t + (|\act(S)|-2)(0.5t - w +1) + t - W -w +2 > |\act(S)|t -w, 
\end{eqnarray*}
\noindent using the fact that $0.5t \geq w$, $t \geq W$, and $|\act(S)| \geq 2$. 
If, on the other hand, $|\act(S)| = 1$, then the LHS of (\ref{equ:2valAlmostTooMuch}) is strictly more than 
\begin{eqnarray*}
1.5t-w \geq t = |\act(S)|t, 
\end{eqnarray*}
and the same also holds for the case when $S$ is an isolated node.
Finally, in the third case, the LHS of (\ref{equ:2valAlmostTooMuch}) is at least 
\begin{eqnarray*}
|\act(S)|(1.5t -W-w+1) + |\act(S)|W > |\act(S)|t. 
\end{eqnarray*}

\item Suppose that $\act(S)$ contains at least two critical nodes, or that $S$ is a cycle and $\act(S)$ has 
at least one critical node. In both cases, $S$ is a bad system. Furthermore, the LHS of (\ref{equ:2valTooMuch}) 
can be lower-bounded by the same calculation as in the previous case with an extra term of $w$. 

\item Suppose that $\act(S)$ contains a hypercritical node. Then the system 
$S$ is bad, and the LHS of (\ref{equ:2valTooMuch}) is at least
\begin{eqnarray*}
(1.5t + 1) + (|\act(S)|-1)(1.5t -W-w+1) + (|\act(S)|-1)W =\\
|\act(S)|t + (|\act(S)|-1)(0.5t-w+1)+0.5t+1 > |\act(S)|t,
\end{eqnarray*}
\noindent where the inequality holds because $0.5t \geq w$. 
\qed
\end{enumerate} 

We argue that Algorithm~1 terminates in polynomial time by the aid of a potential function, defined as 

$$ \Phi = \sum_{v \in \act} (|V|-\level(v)) \cdot (\mbox{number of pebbles at $v$)}.$$ 

Trivially, $0 \leq \Phi \leq |V|\cdot|\pebble|$. The next lemma implies that $\Phi$ is monotonically decreasing after each \textsc{Push} 
operation. 

\begin{lemma} For each node $v \in V$, let $\level(v)$ and $\level'(v)$ denote the levels before and after a 
\textsc{Push} operation, respectively. Then $\level'(v) \geq \level(v)$.

\label{lem:simplePush}
\end{lemma}
\begin{proof}
We prove by contradiction. Suppose that there exist nodes $x$ with $\level'(x) < \level(x)$. Choose 
$v$ to be one among them with minimum $\level'(v)$. By the choice of $v$, and Definition~\ref{def:simplePush}.3, $\level'(v)>0$ and $v\in \act$ after the \textsc{Push} operation. Thus, by Proposition~\ref{pro:first}.4, there exists a node $u$ with 
$\level'(u) = \level'(v) -1$, so that after \textsc{Push}, 

\begin{itemize}
\item Case 1: $v$ is reachable from $u\in\act$, or 
\item Case 2: there exists another node $v'\in\act$ reachable from $u\in\act$ with $\level'(v') = \level'(v)$
in the same system as $v$, and $v$ is critical.
\end{itemize}

Notice that by the choice of $v$, in both cases, $\level'(u) \geq \level(u)$, and $u\in\act$ also before the \textsc{Push} operation. Let $p$ 
be the pebble by which $u$ reaches $v$ (Case 1), or $v'$ (Case 2), after \textsc{Push}. 
Before the \textsc{Push} operation, $p$ was at some node $u'\in\act$ ($u'$ may be $u$, or $p$ is the pebble pushed: from $u'$ to $u$).

By Proposition~\ref{pro:first}.2, in Case 1, $\level(v) \leq \level(u')+1$ (as $v$ is reachable from $u'$ via $p$ before \textsc{Push}), and $\level(v') \leq \level(u')+1$ in Case 2. Furthermore, if in Case 2 $v$ was already critical before \textsc{push}, then $\level(v) \leq \level(v')$ by Proposition~\ref{pro:first}.3 (note that $v'\in\act$ as it is reachable from $u'\in\act$). Hence, in both cases we would have 

$$ \level(v) \leq \level(u')+1 \leq \level(u)+1 \leq \level'(u)+1 = \level'(v),$$ 

\noindent a contradiction. Note that the second inequality holds no matter $u=u'$ or not. 

Finally consider Case 2 where $v$ was not critical before the \textsc{Push} operation. Then a pebble $p'\neq p$ is pushed into $v$ 
in the operation. Note that in this situation, $v$'s system is a tree and contains no critical nodes before \textsc{Push} (by Definition~\ref{def:simplePush}.3); in particular $v'$ is not critical. Furthermore, the presence of $p$ in $u$ implies that $\level(v') \leq \level(u)+1$ by Proposition~\ref{pro:first}.2, and that $v'\in\act$ by Proposition~\ref{pro:first}.1. As $v'$ is not critical, $\level(v') > 0$, and by Proposition~\ref{pro:first}.4 there
exists a node $u''$ with $\level(u'')= \level(v')-1$ so that $u''$ can reach $v'$ by a pebble 
$p''$ ($u''$ may be $u$ and $p''$ may be $p$). As 

$$ \level(v') \leq \level(u)+1 \leq \level'(u)+1 = \level'(v) < \level(v),$$ 

\noindent the \textsc{Push} operation should have pushed $p''$ into $v'$ instead of $p'$ into $v$ (see Definition~\ref{def:simplePush}.4), since $u''$ and $v'$ satisfy all the first three conditions 
of Definition~\ref{def:simplePush}. 
\qed\end{proof}

By Lemma~\ref{lem:simplePush} and the fact that a pebble is pushed to a node with higher level, the potential $\Phi$ strictly decreases after each \textsc{Push} operation, implying that Algorithm~1 finishes 
in polynomial time.

\textbf{Approximation Ratio}: When $t < 2W$, we apply Algorithm~1. 
In the case of $t\geq 2W$, we apply the algorithm of Gairing et al.~\cite{gairing04}, which either correctly reports that 
$\OPT \geq t+1$, or returns an assignment with makespan at most $t+W-1<1.5t$. 

Suppose that $t$ is the smallest number for which an assignment is returned. Then $\OPT \geq t$, and our approximation ratio is bounded by $\frac{1.5t}{\OPT}\leq 1.5$.
We use a theorem to conclude this section. 

\begin{theorem} With arbitrary dedicated loads on the machines, jobs of weight $W$ 
that can be assigned to two machines, and jobs of weight $w$ that can be assigned to any number of machines, 
we can find a $1.5$ approximate solution in polynomial time. 
\label{thm:firstTheorem}
\end{theorem}

In the appendix, we show that a slight modification of our algorithm yields an improved approximation ratio of $1 + \frac{\lfloor \frac{W}{2} \rfloor}{W}$ if $W \geq 2w$.

\section{The General Case}

\label{sec:generalWeights} 

In this section, we describe the core algorithm for the case of arbitrary job weights. This algorithm inherits some basic ideas 
from the previous section, but has several significantly new ingredients---mainly due to the fact that the rocks now have 
different weights. Before formally presenting the algorithm, let us build up intuition by looking at some examples. 

For simplicity, we rescale the numbers and assume that $t=W=1$ and $\beta=0.7$. We aim for an assignment with 
makespan of at most $5/3 + 0.7/3 = 1.9$ or decide that $\OPT > 1$. Consider the example in Figure~\ref{fig:example1}. 
Note that there are $2k+1$ (for some large $k$) nodes (the pattern of the last two nodes repeats). 
Due to node 1 (which can be regarded as the analog of a critical node in the previous section), all edges are to be directed toward the right if we shoot 
for the makespan of 1.9. Suppose that there is an isolated node with the pebble load of $2+\epsilon$ (this node can be regarded as a bad system by itself) 
and it has a pebble of weight 0.7 that can be assigned to node 3, 5, 7 and so on up to $2k+1$. 
Clearly, we do not want to push the pebble into any of them, as it would cause the makespan to be larger than 1.9 by whatever 
orientation. Rather, we should activate node 1 
and send its pebbles away with the aim of relieving the ``congestion'' in the current system (later we will see that 
this is activation rule 1). 
In this example, all odd-numbered nodes are 
activated, and the entire set of nodes (including even-numbered nodes) 
form a \emph{conflict set} (which will be defined formally later). Roughly speaking, the conflict sets
contain activated nodes and the nodes that can be reached by ``backtracking'' the directed 
edges from them. These conflict sets embody the ``congestion'' in the systems. 

\begin{figure}[h]
\centerline{\includegraphics[width=.55\linewidth]{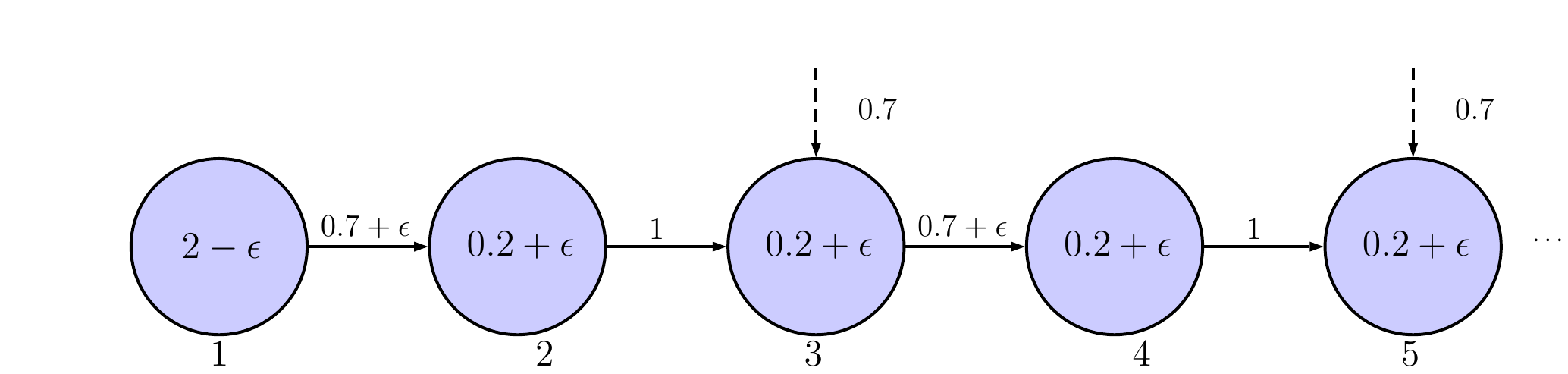}}
\caption{\small There are $2k+1$ nodes (the rest is repeating the same pattern). Numbers inside the shaded circles (nodes) are their pebble load.}
\label{fig:example1}
\end{figure}
\begin{figure}
\centering
\begin{minipage}{.5\textwidth}
\centering
\includegraphics[width=.5\linewidth]{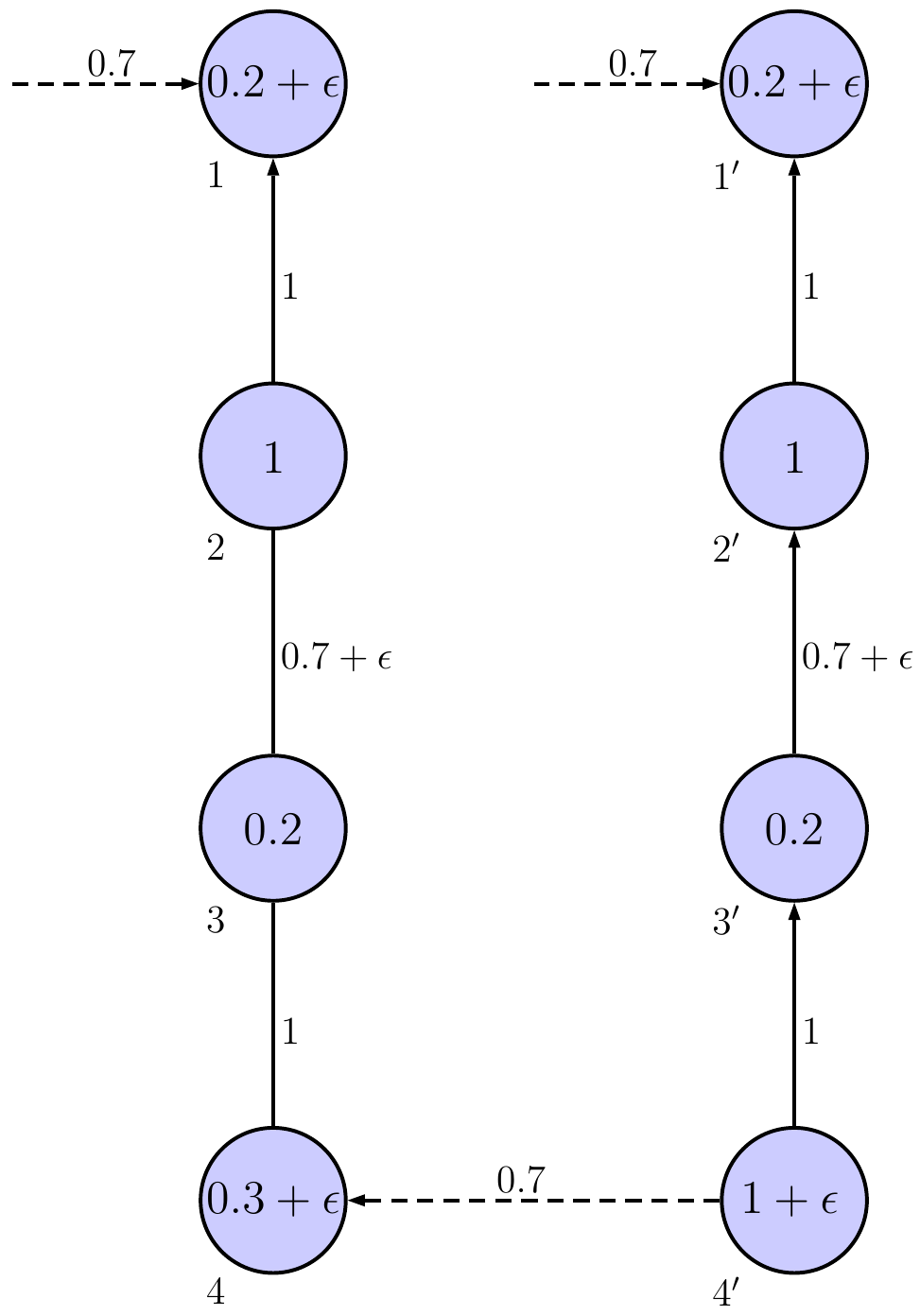}
\caption{A naive \textsc{Push} will oscillate the pebble between nodes $4$ and $4'$.}
\label{fig:example2}
\end{minipage}%
\begin{minipage}{.5\textwidth}
\centering
\includegraphics[width=.5\linewidth]{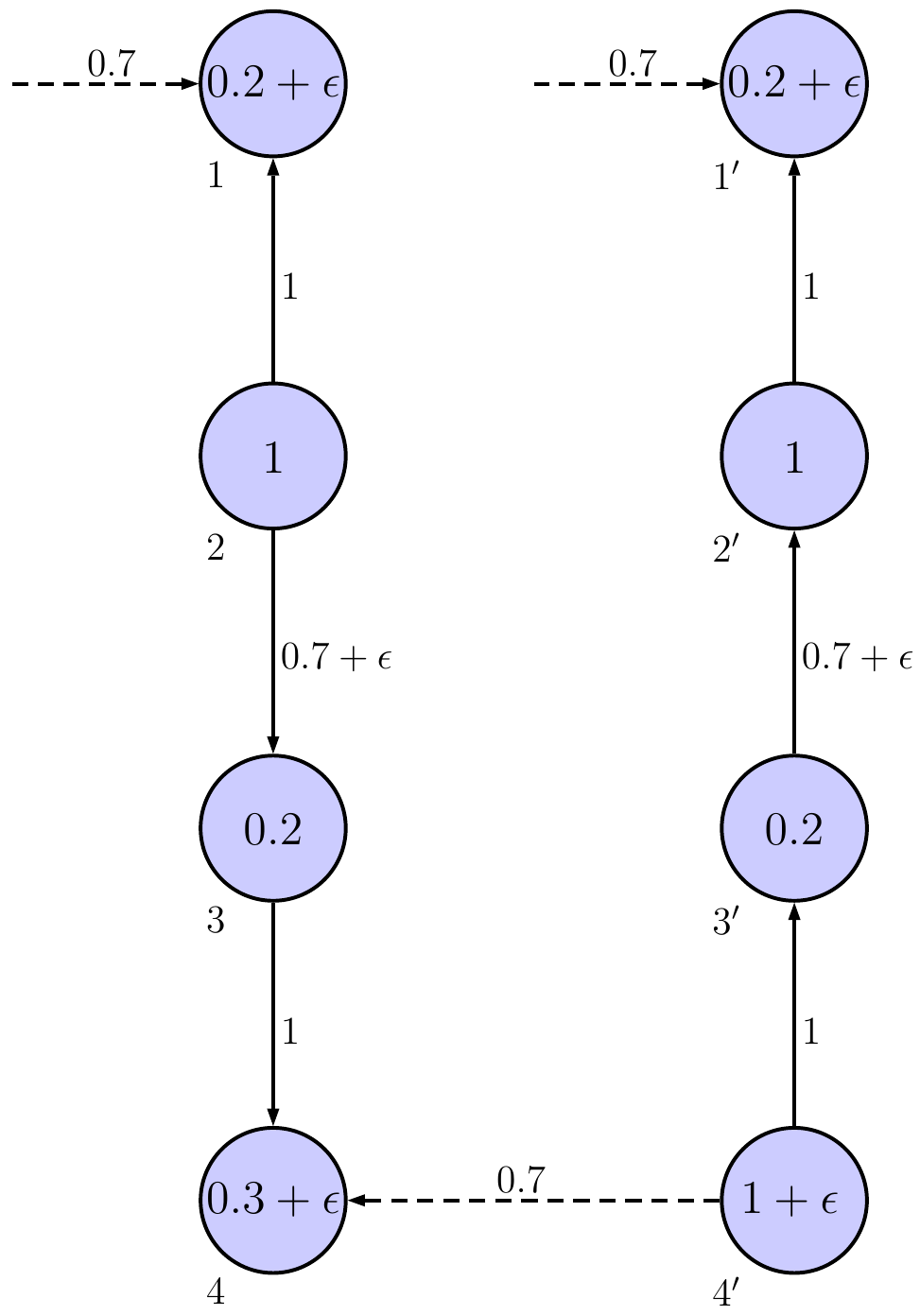}
\caption{A fake orientation from node 2 to 3 causes node 4 to have an incoming edge, thus informing node $4'$ not to push the pebble.}
\label{fig:example3}
\end{minipage}
\end{figure}

Recall that in the previous section, if the \textsc{Push} operation was no longer possible, we argued that the total load 
is too much 
(see the proof of Lemma~\ref{lem:firstContradiction}) for the activated nodes \emph{system by system}. 
Analogously, in this example, we need to argue that in all feasible orientations, the activated set of nodes (totally $k+1$ of them) in this conflict set
cannot handle the total load. However, if all edges are directed toward the left, 
their total load is only $(0.2+\epsilon)k + (2-\epsilon) + (0.7+\epsilon)k = 2+ 0.9k + \epsilon(2k-1)$, which 
is less than what they can handle (which is $k+1$) when $k$ is large. As a result, 
we are unable to arrive at a contradiction. 

To overcome this issue, we introduce another activation rule to strengthen our contradiction argument. 
If all edges are directed to the left, \emph{on the average}, each activated node has a total load of about $0.2+ 0.7$. However, 
each inactivated node has, \emph{on the average}, a total load of about $0.2 + 1$. This motivates 
our activation rule 2 : if an activated node is connected by a ``relatively light'' edge to some 
other node in the conflict set, the latter should be activated as well. The intuition behind is that the two nodes \emph{together}
will receive a relatively heavy load. We remark that it is easy to modify this example to show that 
if we do not apply activation rule 2, then we cannot hope for a $2-\delta$ approximation for any small $\delta>0$.
\footnote{Looking at this particular example, one is tempted to use the idea of activating all nodes in the conflict set. However, such an activation rule will not work. Consider the following example: There are $k+2$ nodes forming a path, and the $k+1$ edges connecting them all have weight $0.95+\epsilon$. The first node has a pebble load of $1$ and thus ``forces'' an orientation of the entire path (for a makespan of at most $1.9$). The next $k$ nodes have a pebble load of $0$, and the last node has a pebble load of $0.25$ and is reachable from a bad system via a pebble of weight $0.7$. The conflict set is the entire path, and activating all nodes leads to a total load of $(k+1)\cdot(0.95+\epsilon)+1+0.25$, which is less than $k+2$ for large $k$.}

Next consider the example in Figure~\ref{fig:example2}. Here nodes $2$, $2'$, and $4'$ can be regarded as the critical nodes, and 
$\{1,2\}$, $\{1',2',3',4'\}$ are the two conflict sets. Both nodes $1$ and $1'$ can be reached by an isolated node with heavy load 
(the bad system) with a pebble of weight 0.7. Suppose further that node $4'$ can reach node $4$ by another pebble of weight $0.7$. 
It is easy to see that a naive \textsc{Push} definition will simply ``oscillate'' the pebble between nodes $4$ and $4'$, causing the algorithm to cycle. 

Intuitively, it is not right to push the pebble from $4'$ into $4$, as it causes the conflict set in the left system to become bigger. Our principle of pushing a pebble 
should be to relieve the congestion in one system, while not worsening the congestion in another. To cope with this problematic case, we use 
\emph{fake orientations}, i.e., we direct edges away from a conflict set, as shown in Figure~\ref{fig:example3}. Node 2 directs the edge toward node 3, which 
in turn causes the next edge to be directed toward node 4. 
With the new incoming edge, node 4 now has a total load of $1+0.3+\epsilon$ to handle, and the pebble thus will not be pushed from node $4'$ to node $4$. 

\subsection{Formal description of the algorithm}
\label{sec:formal}

We inherit some terminology from the previous section. 
We say that $v$ is \emph{reachable} from $u$ if a pebble in $u$ can be assigned to $v$, and that $v$ is reachable from $\act$ if $v$ is reachable from any node $u \in \act$. 
Each tree, cycle, isolated node in $G_{\rock}$ is a system. Note that there is exactly one edge between two adjacent nodes in $G_{\rock}$ (see Proposition~\ref{prop:preprocess}). 
For ease of presentation, we use the short hand $vu$ to refer to the edge $\{v,u\}$ in $G_{\rock}$ and $w_{vu}$ is its weight. 

The orientation of the edges in $G_{\rock}$ will be decided dynamically. 
If $uv$ is directed toward $v$, we call $v$ a \emph{father} of $u$, and $u$ a \emph{child} of $v$ (notice that a node can have several fathers and children). 
We write $rl(v)$ to denote total weight 
of the rocks that are (currently) oriented towards $v$, and $pl(v)$ still denotes the total weight of the pebbles at $v$. An edge that is currently un-oriented is \emph{neutral}. 
In the beginning, all edges in $G_{\rock}$ are neutral. 

A set $\C$ of nodes, called the \emph{conflict set}, will be collected in the course of the algorithm. 
Let $\child(v):=\{u\in\C: \mbox{$u$ is child of $v$}\}$ and $\father(v):=\{u\in\C: \mbox{$u$ is father of $v$}\}$ for any $v\in \C$.
A node $v\in\C$ is a \emph{leaf} if $\child(v)=\emptyset$, and a \emph{root} if $\father(v)=\emptyset$. 
Furthermore, a node $v$ is \emph{overloaded} if $dl(v)+pl(v)+rl(v) > (5/3+\beta/3)t$, and a node $v \in \C$ is \emph{critical} if there exists $u \in \mathcal{F}(v)$ such that $dl(v) + pl(v) + w_{vu} > (5/3 + \beta/3)t$. In other words, a node in the conflict set is critical if it has enough load by itself (without considering incoming rocks) to ``force'' an incident edge to be directed toward a father in the conflict set. 

Initially, the pebbles are arbitrarily assigned to the nodes. The orientation of a subset of the edges 
in $G_{\rock}$ is determined by the procedure \textsc{Forced Orientations} in Figure~\ref{fig:ForcedOrientations}. 

\begin{figure}[!htb]
\centering
\fbox{ \begin{minipage}{0.88\linewidth}
\textsc{Forced Orientations}\\[0.5em]
\textbf{While} $\exists$ neutral edge $vu$ in $G_{\rock}$, s.t. $dl(v)+pl(v)+rl(v)+w_{vu} > (5/3+\beta/3)t$:\\[0.2em] 
\mbox{\hspace*{1em}} \textbf{Direct} $vu$ towards $u$; $\textsc{Marked} := \{u\}$.\\[0.2em] 
\mbox{\hspace*{1em}} \textbf{While} $\exists$ neutral edge $v'u'$ in $G_{\rock}$, s.t. $dl(v')+pl(v')+rl(v')+w_{v'u'} > (5/3+\beta/3)t$ \\[0.2em]
\mbox{\hspace*{1em}} and $v'\in\textsc{Marked}$:\\[0.2em] 
\mbox{\hspace*{2em}} \textbf{Direct} $v'u'$ towards $u'$; $\textsc{Marked} := \textsc{Marked} \cup \{u'\}$.
\end{minipage} }
\caption{The procedure \textsc{Forced Orientations}.}
\label{fig:ForcedOrientations}
\end{figure}

Intuitively, the procedure first finds a ``source node'' $v$, whose dedicated, pebble, and rock load is so high that it ``forces'' an incident edge $vu$ to be oriented away from $v$. The orientation of this edge then propagates through the graph, i.e. edge-orientations induced by the direction of $vu$ are established. Then the next ``source'' is found, and so on. To simplify our proofs, we assume that ties are broken according to a fixed total order if several pairs $(v,u)$ satisfy the conditions of the while-loops. 

The following lemma describes a basic property of the procedure \textsc{Forced Orientations}, that will be used in the subsequent discussion.
\begin{lemma}
Suppose that a node $v$ becomes overloaded during \textsc{Forced Orientations}. Then there exists a path $u_0 u_1 \dots u_k v$ of neutral edges, such that $dl(u_0)+pl(u_0)+rl(u_0)+w_{u_0 u_1} > (5/3+\beta/3)t$ before the procedure, that becomes directed from $u_0$ towards $v$ during the procedure (note that $u_0$ could be $v$). Furthermore, other than $u_k v$, no edge becomes directed toward $v$ in the procedure.
\label{lem:forcedOrientationLemma}
\end{lemma}

\begin{proof}
We start with a simple observation. Let $ab$ be the first edge directed in some iteration of the procedure's outer while-loop; suppose from $a$ to $b$. It is easy to see that up to this moment, no edge has been directed toward $a$ in course of the procedure. Furthermore, if another edge $a'b'$ is directed in the same iteration of the outer while-loop, then there exists a path of neutral edges, starting with $ab$ and ending with $a'b'$, that becomes directed during this iteration. This proves the first part of the lemma. 

Now suppose that some node $v$ becomes overloaded and has more than one edge directed towards it during the procedure. Let $vx$ and $vy$ be the last two edges directed toward $v$, and note that both, $vx$ and $vy$, become directed in the same iteration of the outer while-loop (because as soon as one of the two is directed toward $v$, the other edge satisfies the conditions of the inner while-loop). Hence, there are two different paths directed towards $v$ (with final edges $vx$ and $vy$, respectively), both of which start with the first edge that becomes directed in this iteration of the outer while-loop. This is not possible, since every system is a tree or a cycle, a contradiction.
\qed\end{proof}

Clearly, if after the procedure \textsc{Forced Orientations} a node $v$ still has a neutral incident edge $vu$, then $dl(v)+pl(v)+rl(v)+w_{vu} \leq (5/3+\beta/3)t$. Now suppose that after the procedure, none of the nodes is overloaded. 
Then orienting the neutral edges in each system in such a way that every node has 
at most one more incoming edge gives us a solution with makespan at most $(5/3+\beta/3)t$. 
So assume the procedure ends with a non-empty set of overloaded nodes.
We then apply the procedure \textsc{Explore2} in Figure~\ref{fig:Explore2}.

\begin{figure}[!htb]
\centering
\fbox{ \begin{minipage}{0.85\linewidth}
\textsc{Explore2}\\[0.5em]
\textbf{Initialize} $\act := \emptyset$; $\C:=\emptyset$; $i:=0$. \textbf{Call} \textsc{Forced Orientations}. \\
\textbf{Repeat:} \\
\mbox{\hspace*{1em}} \textbf{If} $i=0$: $\act_i := \{v|v \mbox{ is overloaded}\}$. \\ 
\mbox{\hspace*{1em}} \textbf{Else} $\act_i := \{v| v \not \in \act, \mbox{$v$ is reachable from $\act_{i-1}$}\}$. \\
\mbox{\hspace*{1em}} \textbf{If} $\act_i = \emptyset$: \textbf{stop}.\\ 
\mbox{\hspace*{1em}} $\C_i := \act_i$; $\act := \act \cup \act_i$; $\C := \C \cup \C_i$.\\

\vspace{-0.6em}\mbox{\hspace*{1em}} (\textit{Conflict set construction})\\
\mbox{\hspace*{1em}} \textbf{While} $\exists v \not\in \C$ with a father $u \in \C$ \textbf{or} $\exists$ neutral $vu$ with $v \in \C$ \textbf{do:}\\
\mbox{\hspace*{2em}} \textbf{While} $\exists v \not\in \C$ with a father $u \in \C$: \\
\mbox{\hspace*{3em}} $\C_i := \C_i \cup \{v\}$; $\C := \C \cup \C_i$.\\
\mbox{\hspace*{2em}} \textbf{If} $\exists$ neutral $vu$ with $v \in \C$: \\
\mbox{\hspace*{3em}} \textbf{Direct} $vu$ towards $u$; \textbf{Call} \textsc{Forced Orientations}.\\

\vspace{-0.6em}\mbox{\hspace*{1em}} (\textit{Activation of nodes})\\
\mbox{\hspace*{1em}} \textbf{While} $\exists v \in \C \setminus \act$ satisfying one of the following conditions:\\[0.2em] 
\mbox{\hspace*{2em}} \textit{Rule 1}: $\exists u \in \mathcal{F}(v)$, such that $dl(v)+pl(v) + w_{vu} > (5/3+\beta/3)t$ \\
\mbox{\hspace*{2em}} \textit{Rule 2}: $\exists u \in \act \cap (\child(v)\cup\father(v))$, such that $w_{vu}< (2/3+\beta/3)t$\\
\mbox{\hspace*{1em}} \textbf{Do:} $\act_i := \act_i \cup \{v\}$; $\act := \act \cup \act_i$.\\

\vspace{-0.6em}\mbox{\hspace*{1em}} $i:=i+1$. 
\end{minipage} }
\caption{The procedure \textsc{Explore2}.}
\label{fig:Explore2}
\end{figure}

Let us elaborate the procedure. In each round, we perform the following three tasks.

\begin{enumerate}
\item Add those nodes reachable from the nodes in $\act_{i-1}$ into $\act_i$ in case of $i>1$; or the 
overloaded nodes into $\act_i$ in case of $i=0$. These nodes will be referred to as Type A nodes. 

\item In the sub-procedure \emph{Conflict set construction}, nodes not in the conflict set and having a directed path to those Type A nodes in $\act_i$ 
are continuously added into the 
conflict set $\C_i$. Furthermore, the earlier mentioned \emph{fake orientations}
are applied: each node $v \in \C_i$, if having an incident neutral edge $vu$, 
direct it toward $u$ and call the procedure \textsc{Forced Orientations}. It may happen that 
in this process, two disjoint nodes in $\C_i$ are now connected by a directed path $P$, then all nodes in $P$ along with all nodes 
having a path leading to $P$ 
are added into $\C_i$ (observe that all these nodes have a directed path to some Type A node in $\act_i$). 
We note that the order of fake orientations does not materially affect the outcome of the algorithm (see Lemma~\ref{lem:nondeterminism}). 

\item In the next sub-procedure \emph{Activation of nodes}, we use two rules to activate extra nodes in $\C \backslash \act$. Rule 1 activates the critical nodes; Rule 2 
activates those nodes whose father or child are already activated and they are connected by an edge of weight less than $(2/3+ \beta/3)t$. 
We will refer to the former as Type B nodes and the latter as Type C nodes. 

\end{enumerate}

Observe that except in the initial call of \textsc{Forced Orientations}, no node ever becomes overloaded in \textsc{Explore2} (by Lemma~\ref{lem:forcedOrientationLemma} and the fact that every system is a tree or a cycle). Let us define 
$\level(v)=i$ if $v \in \act_{i}$. In case $v \not \in \act$, let $\level(v) =\infty$. 
The next proposition summarizes some important properties of the procedure \textsc{Explore2}. 

\begin{proposition}After the procedure \textsc{Explore2}, the following holds. 
\begin{enumerate}
\item\label{item:Aclosed} All nodes reachable from $\act$ are in $\act$. 
\item Suppose that $v \in \act$ is reachable from $u \in \act$. Then $\level(v) \leq \level(u)+1$. 

\end{enumerate}

Furthermore, at the end of each round $i$, the following holds. 

\begin{enumerate}\setcounter{enumi}{2}

\item Every node $v$ that can follow a directed path to a node in $\C := \cup_{\tau=0}^{i} \C_\tau$ is in $\C$. 
Furthermore, 
if a node $v \in \C$ has an incident edge $vu$ with $u \not \in \C$, then $vu$ is directed toward $u$. 

\item Each node $v \in \act_i$ is one of the following three types.

\begin{enumerate}
\item \emph{\textbf{Type A}}: there exists another node $u \in \act_{i-1}$ so that $v$ is reachable from $u$, or 
$v$ is overloaded and is part of $\act_0$. 
\item \emph{\textbf{Type B}}: $v$ is activated via Rule 1 (hence $v$ is critical)\footnote{For simplicity, if a node can be activated by both Rule 1 and Rule 2, we assume it is activated by Rule 1.}, and there exists a directed path from $v$ to $u \in A_{i}$ of Type A. 
\item \emph{\textbf{Type C}}: $v$ is activated via Rule 2, and there exists an adjacent node $u \in \cup_{\tau=0}^{i} \act_\tau$ so that 
$w_{vu} < (2/3+\beta/3)t$ and $u \in \child(v)\cup\father(v)$. 
\end{enumerate} 
\end{enumerate}
\label{pro:third}
\end{proposition}

After the procedure \textsc{Explore2}, we apply the \textsc{Push} operation (if possible), defined as follows. 

\begin{definition} \textsc{Push} operation: push a pebble from $u^*$ to $v^*$ if the following conditions hold
(if there are multiple candidates, pick any). 
\label{def:complicatedPush}
\begin{enumerate}

\item The pebble is at $u^*$ and it can be assigned to $v^*$. 
\item $\level(v^*)=\level(u^*)+1$. 
\item $dl(v^*)+pl(v^*)+rl(v^*) \leq (5/3-2/3\cdot\beta)t$. 
\item $\child(v^*)=\emptyset$, or $dl(v^*)+pl(v^*) + w_{v^*u} \leq (5/3-2/3\cdot\beta)t$ for all $u\in\father(v)$. 
\end{enumerate}

\end{definition}

Definition~\ref{def:complicatedPush}(3) is meant to make sure that $v^*$ does not become overloaded after receiving 
a new pebble (whose weight can be as heavy as $\beta t$). Definition~\ref{def:complicatedPush}(4) 
says either $v^*$ is a leaf, or adding a pebble with weight 
as heavy as $\beta t$ does not cause $v^*$ to become critical. 

\begin{quote} \textbf{Algorithm~2}: Apply \textsc{Explore2}. If it ends with $\act_0=\emptyset$, return a solution with makespan at most $(5/3+\beta/3)t$. 
Otherwise, apply \textsc{Push}. If \textsc{push} is impossible, declare that $\OPT \geq t+1$. Un-orient all edges 
in $G_{\rock}$ and repeat this process. 

\end{quote}

\begin{lemma} When there is at least one overloaded node and the \textsc{Push} operation is no longer possible, $\OPT \geq t+1$. 
\label{lem:generalContradiction}
\end{lemma}

\begin{lemma} For each node $v \in V$, let $\level(v)$ and $\level'(v)$ denote the levels before and after a 
\textsc{Push} operation, respectively. Then $\level'(v) \geq \level(v)$.
\label{lem:generalPush}
\end{lemma}

The preceding two lemmas are proven in sections~\ref{sec:LemGeneralContradiction} and~\ref{sec:LemGeneralPush}, respectively. We again use the potential function

$$ \Phi = \sum_{v \in \act} (|V|-\level(v)) \cdot (\mbox{number of pebbles at $v$)}$$ 

to argue the polynomial running time of Algorithm~2. Trivially, $0 \leq \Phi \leq |V|\cdot|\pebble|$. Furthermore, by Lemma~\ref{lem:generalPush} and the fact that a pebble is pushed to a node with higher level, the potential $\Phi$ strictly decreases after each \textsc{Push} operation. This implies that Algorithm~2 finishes in polynomial time. 

We can therefore conclude:

\begin{theorem} Let $\beta \in [4/7, 1)$. With arbitrary dedicated loads on the machines, if jobs of weight greater than $\beta W$
can be assigned to only two machines, and jobs of weight at most $\beta W$ can be assigned to any number of machines, 
we can find a $5/3+\beta/3$ approximate solution in polynomial time. 
\label{thm:secondTheorem}
\end{theorem}

\subsection{Proof of Lemma~\ref{lem:generalContradiction}}\label{sec:LemGeneralContradiction}
Our goal is to show that in any feasible solution, the activated nodes $\act$ must handle a total load of more than $|\act| t$,
which implies that $\OPT \geq t+1$. For the proof, we focus on a single component $K$ of $G_\rock[\C]$, the subgraph of $G_\rock$ induced by the conflict set $\C$, 
and a fixed orientation $\psi \in \allorient$. Let $\psi(v)$ denote the total weight of the rocks assigned to 
any $v\in \act$ by $\psi$ (note that $0\leq \psi(v) \leq t$), and let $\act(K)$ denote the set of activated nodes in $K$. We will show that

\begin{equation}\label{eqn:generalContra}
\sum_{v \in \act(K)} pl(v) + dl(v) + \psi(v)> |\act(K)|t
\end{equation}

if $\act(K)\neq \emptyset$. The lemma then follows by summing over all components of $G_\rock[\C]$, and noting that the pebbles on the nodes in $\act$ can only be assigned to the nodes in $\act$ (Proposition \ref{pro:third}(\ref{item:Aclosed})).

If $K$ consists only of a single activated node $v$, then (\ref{eqn:generalContra}) clearly holds, as $pl(v) + dl(v) > (5/3-2/3\cdot\beta)t \geq t$ (since $v$ is a Type A node and \textsc{push} is no longer possible). In the following, we will assume that $\father(v)\cup\child(v)\neq\emptyset$ for all $v\in\act(K)$.

\begin{definition} 
For every non-leaf $v\in\act(K)$, fix some node $d(v)\in\child(v)$, such that $w_{v d(v)} = \max_{u \in \child(v)}w_{vu}$. 
\end{definition}
\begin{definition} 
For every non-root $v\in\act(K)$, fix some node $f(v)\in\father(v)$, such that $w_{v f(v)} = \max_{u \in \father(v)}w_{vu}$.
\end{definition}
\begin{definition} 
For every node $v\in\act(K)$ that is neither a root nor a leaf, fix some node $n(v)\in\child(v)\cup\father(v)$, such that $w_{v n(v)} = \max_{u \in\child(v)\cup\father(v)}w_{vu}$.
\end{definition}
\begin{definition}
For every node $v\in\act(K)$ that was activated using Rule 2 in the final execution of \textsc{Explore2}, fix some node $a(v)\in\act(K)\cap(\child(v)\cup\father(v))$ 
with $w_{v a(v)}<(2/3+\beta/3)t$, such that $a(v)$ has been activated before $v$.

\end{definition}

We classify the nodes $v \in \act(K)$ that are neither a root nor a leaf, into the following three types. 
\begin{description}
\item[Type 1:] $|\child(v)|>1$.
\item[Type 2:] $|\child(v)|=1$ and $v$ was activated via Rule 2 (i.e., as a Type C node). 
\item[Type 3:] $|\child(v)|=1$ and $v$ was not activated via Rule 2 (i.e. as a Type A or Type B node). 

\end{description}

In the following, we summarize the inequalities that we use for the different types of nodes, in order to prove~\eqref{eqn:generalContra}. We refer to them as the \emph{load-inequalities}.

\begin{numberedClaim} For every leaf $v \in \act(K)$, $pl(v)+dl(v) > (5/3+\beta/3)t - w_{v f(v)}$.
\label{cla:plLeaves}
\end{numberedClaim}

\begin{proof}

If $v \in \act_i$ is activated as a Type A node, then it is either overloaded or is reachable from a node $u \in \act_{i-1}$. 
In both cases, since \textsc{push} is no longer possible, $pl(v)+dl(v)+rl(v) > (5/3-2/3\cdot\beta)t$. The claim follows as $rl(v)=0$ and $w_{v f(v)} > \beta t$. 
If $v$ is not activated as a Type A node, then $v$ first becomes part of $\C$ and then becomes activated via Rule 1 or Rule 2. 
In this case, at the moment $v$ becomes part of $\C$, it must have a father $u \in \C$. The edge $vu$ becomes oriented towards $u$ 
only when \textsc{Forced Orientations} is called and $dl(v)+pl(v)+rl(v)+w_{vu} > (5/3+\beta/3)t$. The claim follows again as $rl(v)=0$ and $w_{v f(v)} \geq w_{vu}$.
\qed\end{proof}

\begin{numberedClaim} For every root $v\in \act(K)$ with $|\child(v)|=1$, $pl(v)+dl(v) > (5/3-2/3\cdot\beta)t - w_{v d(v)}$.
\label{cla:plRootsA}
\end{numberedClaim}

\begin{proof}
As $v \in \act_i$ has no father in $\C$, it must either be overloaded or reachable from an activated node $u \in \act_{i-1}$. 
In both cases, $pl(v)+dl(v)+rl(v) > (5/3-2/3\cdot\beta)t$, since the \textsc{Push} operation is no longer possible. The claim follows as $|\child(v)|=1$ implies $w_{v d(v)} \geq rl(v)$.
\qed\end{proof}

\begin{numberedClaim} For every root $v\in \act(K)$ with $|\child(v)|>1$, $pl(v)+dl(v) \geq 0$.
\label{cla:plRootsB}
\end{numberedClaim}
\begin{proof}
Trivially true.
\qed\end{proof}
\begin{numberedClaim} For every Type 1 node $v\in \act(K)$, $pl(v)+dl(v) \geq 0$.
\label{cla:plType1}
\end{numberedClaim}
\begin{proof}
Trivially true.
\qed\end{proof}

\begin{numberedClaim} For every Type 2 node $v\in \act(K)$, $pl(v)+dl(v) > (5/3+\beta/3)t - w_{v f(v)} - w_{v d(v)}$.
\label{cla:plType2}
\end{numberedClaim}
\begin{proof}
As $v$ is activated using Rule 2, it first becomes part of $\C$ without being activated. For this to happen, 
it must have a father $u \in \C$. The edge $vu$ becomes oriented towards $u$ only when \textsc{Forced Orientations} is called and $dl(v)+pl(v)+rl(v)+w_{vu} > (5/3+\beta/3)t$. 
The claim follows as $w_{v d(v)} \geq rl(v)$ (since $|\child(v)|=1$) and $w_{v f(v)} \geq w_{vu}$.
\qed\end{proof}
\begin{numberedClaim} For every Type 3 node $v\in \act(K)$, $pl(v)+dl(v) > (5/3-2/3\cdot\beta)t - w_{v n(v)}$.
\label{cla:plType3}
\end{numberedClaim}

\begin{proof}
If $v$ is overloaded, the claim directly follows from the fact that $w_{v n(v)}\geq rl(v)$. Furthermore, if $v \in A_{i}$ is reachable from an activated node $u\in A_{i-1}$, 
then the claim follows from the definition of $n(v)$ and the fact that either the third or the fourth condition of \textsc{push} must be violated. 
The only other possibility for $v$ to be activated is via Rule 1, which together with the definition of $n(v)$ implies our claim. 
\qed\end{proof}

To prove~\eqref{eqn:generalContra}, we look at each node $v\in \act(K)$ separately and calculate how much it contributes to the balance under some simplifying assumptions. In the end, we will see that the nodes in $\act(K)$ have enough load to compensate for the assumptions we made.

Let $E_{\act(K)}$ denote the edges of $K$ that are incident with the nodes $\act(K)$, i.e. $E_{\act(K)}:=\{vu \in \rock : u \in \act(K), v \in \child(u)\cup\father(u)\}$. We say that an edge $vu\in E_{\act(K)}$ is \emph{covered} if $w_{vu}$ appears on the right-hand side of $u$'s and/or $v$'s load-inequality. For example, if $v$ is a leaf, then $v f(v)$ is covered. Every edge in $E_{\act(K)}$ that is not covered is called \emph{uncovered}. Finally, we say that an edge $vu\in E_{\act(K)}$ is \emph{doubly covered} if $w_{vu}$ appears on the right-hand side of both $u$'s and $v$'s load-inequality.

We distinguish two cases.

\subsubsection{Case 1: $K$ is a tree.}

\begin{numberedClaim} $K$ contains $1+\sum_{v \in K: \father(v)\neq\emptyset}(|\father(v)|-1)$ many roots, and 
$1+\sum_{v \in K: \child(v)\neq\emptyset}(|\child(v)|-1)$ many leaves. Furthermore, every root and leaf in $K$ is activated.
\label{cla:RootsLeavesTree}
\end{numberedClaim}
\begin{proof}
The first part simply follows from the degree sum formula for directed graphs and the fact that $K$ is a tree. For the second part, observe that any node $v\in \C$ that is not activated as Type A node, must have had a father $u \in \C$ already before it got added into $\C$ itself. This proves that every root in $K$ is activated (as a Type A node). 

If a leaf $v\in \C$ is not activated as Type A node, then its incident edge $vu$ with $u \in \C$ is oriented toward $u$ only when 
\textsc{Forced Orientations} is called and $dl(v)+pl(v)+rl(v)+w_{vu} > (5/3+\beta/3)t$. As $v\in \C$ ends up a leaf, $rl(v)=0$, and Rule 1 would have applied to $v$. So every leaf in $K$ is activated. 
\qed\end{proof}

In our calculations, we will assume that every covered edge $vu\in E_{\act(K)}$ has weight $w_{vu}=t$, and that $\psi(v)=t$ for all $v \in\act(K)$. With these assumptions, we will show that 
\begin{equation}\label{eqn:StrongerContra}
\begin{split}
\sum_{v \in \act(K)} pl(v) &+ dl(v) + \psi(v)> |\act(K)|t \\ &-|\{\text{doubly covered }vu \in E_{\act(K)} : w_{vu}<(2/3+\beta/3)t\}|\cdot (1/3-\beta/3)t\\ &+ |\{ \text{uncovered }vu \in E_{\act(K)}\}|\cdot (t-w_{vu})\\ &+t. 
\end{split}
\end{equation}

Let us consider the error caused by these two assumptions when we lower-bound the term $\sum_{v \in \act(K)} pl(v) + dl(v) + \psi(v)$, and in doing so, we will show 
why \eqref{eqn:StrongerContra} implies~\eqref{eqn:generalContra}.

Consider an edge $vu \in E_{\act(K)}$ that $\psi$ assigns to a node in $\act(K)$, say $v$. Consider three possibilities. 

\begin{itemize} 
\item If $vu$ is covered, then $w_{vu}$ appears on the LHS of~\eqref{eqn:generalContra} 
as a negative term after we plug in the load-inequalities, and the two terms $\psi(v)$ and $w_{vu}$ cancel each other. Hence, in this case, we make no error by assuming both terms to be equal to $t$. 
\item If $vu$ is doubly covered and $w_{vu}<(2/3+\beta/3)t$, our assumptions underestimate the load $\sum_{v \in \act(K)} pl(v) + dl(v) + \psi(v)$ by more than $(1/3-\beta/3)t$. 

\item If $vu$ is uncovered, then we overestimate $\psi(v)$ by at most $t-w_{vu}$. 

\end{itemize} 

Finally, we note that $\psi$ must assign an edge from $E_{\act(K)}$ to every node in $\act(K)$ except for possibly one. For this special node $v^*$ that does not receive an edge 
from $E_{\act(K)}$ under $\psi$, we overestimate $\psi(v^*)$ by at most $t$.
In conclusion, when we remove our assumptions, $\sum_{v \in \act(K)} pl(v) + dl(v) + \psi(v)$ increases by more than $(1/3-\beta/3)t$ per doubly covered edge $vu \in E_{\act(K)}$ with $w_{vu}<(2/3+\beta/3)t$, and decreases by at most $t-w_{vu}$ per uncovered edge $vu \in E_{\act(K)}$, plus possibly another $t$ for the special node $v^*$. Hence, if we prove inequality~\eqref{eqn:StrongerContra} under the aforementioned assumptions, \eqref{eqn:generalContra} must hold after we remove the assumptions, and Lemma~\ref{lem:generalContradiction} 
would follow. 

We now turn to proving~\eqref{eqn:StrongerContra} when every covered edge $vu\in E_{\act(K)}$ has weight $w_{vu}=t$, and $\psi(v)=t$ for all $v \in\act(K)$. To this end, we consider the value $pl(v) + dl(v) + \psi(v)$ as a \emph{budget} of node $v$. Furthermore, we also assign budgets to edges $vu\in E_{\act(K)}$ that are doubly covered and have weight $w_{vu}<(2/3+\beta/3)t$. Each of them gets a budget of $(1/3-\beta/3)t$. Other remaining edges of $E_{\act(K)}$ have budget 0. 

By redistributing budgets between nodes and edges, we will ensure that eventually 

\begin{enumerate}
\item[(i)] every node in $\act(K)$ has a budget of at least $t$,

\item[(ii)] there exists a leaf in $\act(K)$ with budget strictly greater than $t+(2/3+\beta/3)t$,

\item[(iii)] there exists a root in $\act(K)$ with budget at least $t+(2/3-2/3\cdot\beta)t$, 

\item [(iv)] every uncovered edge $vu \in E_{\act(K)}$ has a budget of at least $t-w_{vu}$, and

\item[(v)] no edge in $E_{\act(K)}$ has negative budget. 

\end{enumerate}

This would complete the proof.

We start with the leaf nodes. If $v \in \act(K)$ is a leaf, then (using Claim~\ref{cla:plLeaves}) it has a budget of more than $(5/3+\beta/3)t - w_{v f(v)} + \psi(v) = (5/3+\beta/3)t$. Using Claim~\ref{cla:RootsLeavesTree}, we can therefore add $(|\child(u)|-1)\cdot (2/3+\beta/3)t$ to the budget of every non-leaf $u\in \act(K)$, such that (i) and (ii) are still satisfied for all leaves.

Next we consider the roots. If $v \in \act(K)$ is a root and $|\child(v)|=1$, then (using Claim~\ref{cla:plRootsA}) it has a budget of more than $(5/3-2/3\cdot\beta)t$. If $v \in \act(K)$ is a root and $|\child(v)|>1$, then (using Claim~\ref{cla:plRootsB} and the load added in the previous step) it has a budget of at least $t+ (|\child(v)|-1)\cdot (2/3+\beta/3)t$. In the latter case, we transfer $(2/3-2/3\cdot\beta)t$ to the budget of every edge in $E_{\act(K)}$ that is incident with $v$. The budget of $v$ thereby remains at least $t+ (|\child(v)|-1)\cdot (2/3+\beta/3)t - |\child(v)|\cdot (2/3-2/3\cdot\beta)t = (1/3-\beta/3)t + |\child(v)|\cdot \beta t\geq (5/3-2/3\cdot\beta)t$, where the last inequality follows from $|\child(v)|\geq 2$ and $\beta \geq 4/7$. Using Claim~\ref{cla:RootsLeavesTree}, we can thus add $(|\father(u)|-1)\cdot (2/3-2/3\cdot\beta)t$ to the budget of every non-root $u\in \act(K)$, such that (i) and (iii) are satisfied for all roots.

Before we move on to Type 1, 2, and 3 nodes, we take one step back and visit the leaves again, as their budget has increased again through the latest redistribution of load. Namely, every leaf $v \in \act(K)$ got an additional load of $(|\father(v)|-1)\cdot (2/3-2/3\cdot\beta)t$, which we now use to add $(2/3-2/3\cdot\beta)t$ to the budget of every edge in $E_{\act(K)}$ that is incident with $v$, except to $v f(v)$ (which is surely covered). After this, (ii) and (iii) are satisfied, (i) holds for every root and every leaf, and every uncovered edge $vu \in E_{\act(K)}$ that is incident with a root or a leaf has a budget of at least $(2/3-2/3\cdot\beta)t$. 

Let us now consider the nodes of Type 1. Such a node $v$ (using Claim~\ref{cla:plType1} and the load added in previous steps) has a budget of at least $t + (|\child(v)|-1)\cdot (2/3+\beta/3)t + (|\father(v)|-1)\cdot (2/3-2/3\cdot\beta)t$. We transfer $(2/3-2/3\cdot\beta)t$ to the budget of every edge in $E_{\act(K)}$ that is incident with $v$. Since there are $|\child(v)|+|\father(v)|$ such edges, the budget at $v$ remains at least $t + (|\child(v)|-1)\cdot (2/3+\beta/3)t - (|\child(v)|+1)\cdot (2/3-2/3\cdot\beta)t=(|\child(v)|+1)\beta t-(1/3+2/3\cdot\beta)t\geq t$, as $|\child(v)|\geq 2$ and $\beta \geq 4/7$.

Next we consider the nodes of Type 2. Such a node $v$ (using Claim~\ref{cla:plType2} and the load added in previous steps) has a budget of more than $(2/3+\beta/3)t+ (|\father(v)|-1)\cdot (2/3-2/3\cdot\beta)t$. We transfer $(2/3-2/3\cdot\beta)t$ to the budget of every edge in $E_{\act(K)}$ that is incident with $v$, except to $v f(v)$ and $v d(v)$ (which are surely covered). Since there are $|\father(v)|-1$ such edges, the resulting budget at $v$ is still more than $(2/3+\beta/3)t$. We now reduce the budget of the edge $v a(v)$ by $(1/3-\beta/3)t$ and add this load to $v$'s budget, which is then more than $t$.
We will show later that this last step (reducing the budget of $va(v)$) does not cause a violation of (v). 

Finally, we consider the nodes of Type 3. Such a node $v$ (using Claim~\ref{cla:plType3} and the load added in previous steps) has a budget of more than $(5/3-2/3\cdot\beta)t+ (|\father(v)|-1)\cdot (2/3-2/3\cdot\beta)t$. We transfer $(2/3-2/3\cdot\beta)t$ to the budget of every edge in $E_{\act(K)}$ that is incident with $v$, except to $v n(v)$ (which is surely covered). Since there are $|\father(v)|$ such edges, the resulting budget at $v$ is still more than $t$.

After the above redistributions of load, (i), (ii), and (iii) are satisfied. Furthermore, suppose that some edge $vu\in E_{\act(K)}$ is uncovered and has weight $w_{vu}\geq (2/3+\beta/3)t$. Then at least once, we have added $(2/3-2/3\cdot\beta)t$ to the budget of this edge, and we never reduced it. Therefore it has a budget of at least $(2/3-2/3\cdot\beta)t\geq (1/3-\beta/3)t \geq t-w_{vu}$, and (iv) holds for this edge. If, on the other hand, an uncovered edge $vu\in E_{\act(K)}$ has weight $w_{vu}< (2/3+\beta/3)t$, then both $u$ and $v$ are in $\act(K)$ (due to activation rule 2), and $(2/3-2/3\cdot\beta)t$ was added twice to the budget of $vu$. Furthermore, if this budget got reduced at some point, then at most once ($u = a(v)$ and $v = a(u)$ cannot happen simultaneously). The final budget of $vu$ is thus at least $2\cdot (2/3-2/3\cdot\beta)t-(1/3-\beta/3)t= t-\beta t>t-w_{vu}$. 
Hence, for such an edge the assertion (iv) also holds. 

Finally, for (v), observe that the only point where we reduce the budget of a covered edge $vu \in E_{\act(K)}$ and add it to $v$'s budget, is when $v$ is of Type 2, $w_{vu}<(2/3+\beta/3)t$, and $u = a(v)$. Furthermore, both $u$ and $v$ have to be in $\act(K)$ (due to activation rule 2). In this case, the budget of $vu$ is reduced exactly once, by a value of $(1/3-\beta/3)t$. If $vu$ is doubly covered, then it had an initial budget of $(1/3-\beta/3)t$, and its budget therefore remains non-negative. If, on the other hand, $vu$ is covered but not doubly covered, then at some point its budget was increased by $(2/3-2/3\cdot\beta)t$. Hence, the final budget is at least $(2/3-2/3\cdot\beta)t-(1/3-\beta/3)t=(1/3-\beta/3)t\geq0$. This concludes the proof.

\subsubsection{Case 2: K is a cycle.}

\begin{numberedClaim} $K$ contains $\sum_{v \in K: \father(v)\neq\emptyset}(|\father(v)|-1)$ many roots, and 
$\sum_{v \in K: \child(v)\neq\emptyset}(|\child(v)|-1)$ many leaves. Furthermore, every root and leaf in $K$ is activated.
\label{cla:RootsLeavesCycle}
\end{numberedClaim}
\begin{proof}
The first part simply follows from the degree sum formula for directed graphs and the fact that $K$ is a cycle. The second part is analogous to Claim~\ref{cla:RootsLeavesTree}.
\qed\end{proof}

We will again assume that every covered edge $vu\in E_{\act(K)}$ has weight $w_{vu}=t$, and that $\psi(v)=t$ for all $v \in\act(K)$. With these assumptions, we will show that 
\begin{equation}\label{eqn:StrongerContraCycle}
\begin{split}
\sum_{v \in \act(K)} pl(v) &+ dl(v) + \psi(v)> |\act(K)|t \\ &-|\{\text{doubly covered }vu \in E_{\act(K)} : w_{vu}<(2/3+\beta/3)t\}|\cdot (1/3-\beta/3)t\\ &+ |\{ \text{uncovered }vu \in E_{\act(K)}\}|\cdot (t-w_{vu}). 
\end{split}
\end{equation}

By the same arguments as in Case 1, the error caused by the above two assumptions when we lower-bound the term $\sum_{v \in \act(K)} pl(v) + dl(v) + \psi(v)$ is: 
\begin{itemize} 
\item we underestimate the term by more than $(1/3-\beta/3)t$ per doubly covered edge $vu \in E_{\act(K)}$ with $w_{vu}<(2/3+\beta/3)t$,
\item we overestimate the term by at most $t-w_{vu}$ per uncovered edge $vu \in E_{\act(K)}$. 
\end{itemize} 
Note that, since $K$ is a cycle, $\psi$ must assign an edge from $E_{\act(K)}$ to every node in $\act(K)$, and thus there is no special node $v^*$ as in Case 1. Hence, if we prove inequality~\eqref{eqn:StrongerContraCycle} under the aforementioned assumptions, \eqref{eqn:generalContra} must hold after we remove the assumptions, and Lemma~\ref{lem:generalContradiction} would follow.

We now prove~\eqref{eqn:StrongerContraCycle} when every covered edge $vu\in E_{\act(K)}$ has weight $w_{vu}=t$, and $\psi(v)=t$ for all $v \in\act(K)$. Again, we consider the value $pl(v) + dl(v) + \psi(v)$ as a \emph{budget} of node $v$. Furthermore, we also assign budgets to edges $vu\in E_{\act(K)}$ that are doubly covered and have weight $w_{vu}<(2/3+\beta/3)t$. Each of them gets a budget of $(1/3-\beta/3)t$. Other remaining edges of $E_{\act(K)}$ have budget 0. 

By redistributing budgets between nodes and edges, we will ensure that eventually 

\begin{enumerate}
\item[(i)] every node in $\act(K)$ has a budget of at least $t$,

\item[(ii)] at least one node in $\act(K)$ has a budget strictly greater than $t$,

\item [(iii)] every uncovered edge $vu \in E_{\act(K)}$ has a budget of at least $t-w_{vu}$, and

\item[(iv)] no edge in $E_{\act(K)}$ has negative budget. 

\end{enumerate}

This would complete the proof.

We start with the leaf nodes. If $v \in \act(K)$ is a leaf, then (using Claim~\ref{cla:plLeaves}) it has a budget of more than $(5/3+\beta/3)t - w_{v f(v)} + \psi(v) = (5/3+\beta/3)t$. Using Claim~\ref{cla:RootsLeavesCycle}, we can therefore add $(|\child(u)|-1)\cdot (2/3+\beta/3)t$ to the budget of every non-leaf $u\in \act(K)$, such that (i) is still satisfied for all leaves.

Next we consider the roots. If $v \in \act(K)$ is a root and $|\child(v)|=1$, then (using Claim~\ref{cla:plRootsA}) it has a budget of more than $(5/3-2/3\cdot\beta)t$. If $v \in \act(K)$ is a root and $|\child(v)|>1$, then (using Claim~\ref{cla:plRootsB} and the load added in the previous step) it has a budget of at least $t+ (|\child(v)|-1)\cdot (2/3+\beta/3)t$. In the latter case, we transfer $(2/3-2/3\cdot\beta)t$ to the budget of every edge in $E_{\act(K)}$ that is incident with $v$. The budget of $v$ thereby remains at least $t+ (|\child(v)|-1)\cdot (2/3+\beta/3)t - |\child(v)|\cdot (2/3-2/3\cdot\beta)t = (1/3-\beta/3)t + |\child(v)|\cdot \beta t\geq (5/3-2/3\cdot\beta)t$, where the last inequality follows from $|\child(v)|\geq 2$ and $\beta \geq 4/7$. Using Claim~\ref{cla:RootsLeavesCycle}, we can thus add $(|\father(u)|-1)\cdot (2/3-2/3\cdot\beta)t$ to the budget of every non-root $u\in \act(K)$, such that (i) is satisfied for all roots.

Before we move on to Type 1, 2, and 3 nodes, we take one step back and visit the leaves again, as their budget has increased again through the latest redistribution of load. Namely, every leaf $v \in \act(K)$ got an additional load of $(|\father(v)|-1)\cdot (2/3-2/3\cdot\beta)t$, which we now use to add $(2/3-2/3\cdot\beta)t$ to the budget of every edge in $E_{\act(K)}$ that is incident with $v$, except to $v f(v)$ (which is surely covered). After this, (i) holds for every root and every leaf, and every uncovered edge $vu \in E_{\act(K)}$ that is incident with a root or a leaf has a budget of at least $(2/3-2/3\cdot\beta)t$. 

As $K$ is a cycle, there cannot be a node of Type 1, since every $v \in \act(K)$ with $|\child(v)|>1$ is a root. 

Let us now consider the nodes of Type 2. Such a node $v$ (using Claim~\ref{cla:plType2} and the load added in previous steps) has a budget of more than $(2/3+\beta/3)t+ (|\father(v)|-1)\cdot (2/3-2/3\cdot\beta)t$. We transfer $(2/3-2/3\cdot\beta)t$ to the budget of every edge in $E_{\act(K)}$ that is incident with $v$, except to $v f(v)$ and $v d(v)$ (which are surely covered). Since there are $|\father(v)|-1$ such edges, the resulting budget at $v$ is still more than $(2/3+\beta/3)t$. We now reduce the budget of the edge $v a(v)$ by $(1/3-\beta/3)t$ and add this load to $v$'s budget, which is then more than $t$.
We will show later that this last step (reducing the budget of $va(v)$) does not cause a violation of (iv). 

Finally, we consider the nodes of Type 3. Such a node $v$ (using Claim~\ref{cla:plType3} and the load added in previous steps) has a budget of more than $(5/3-2/3\cdot\beta)t+ (|\father(v)|-1)\cdot (2/3-2/3\cdot\beta)t$. We transfer $(2/3-2/3\cdot\beta)t$ to the budget of every edge in $E_{\act(K)}$ that is incident with $v$, except to $v n(v)$ (which is surely covered). Since there are $|\father(v)|$ such edges, the resulting budget at $v$ is still more than $t$.

After the above redistributions of load, (i) is satisfied. Furthermore, (ii) holds as at least one node must be of Type 2, Type 3, or a leaf, and for all these cases the load-inequality is a strict inequality. Finally, the proof of (iii) and (iv) is exactly analogous to the proof of (iv) and (v) in Case 1. 

\subsection{Proof of Lemma~\ref{lem:generalPush}}\label{sec:LemGeneralPush}

In the following, let $E(V')$ denote the set of edges both of whose endpoints are in $V'$ and 
$\delta(V')$ the set of edges exactly one of whose endpoints is in $V'$, for each $V' \subseteq V$. 

We prove the lemma by the following two steps. \\

\noindent \textbf{Step 1}: We create a clone of the pebble that is pushed from $u^*$ to $v^*$ and put this cloned pebble at $v^*$ (by cloning, 
we mean the new pebble has the same weight and the same set of machines it can be assigned to) and keep the old one at $u^*$. 
We apply 
\textsc{Explore2} to this new instance and argue that the outcome is ``essentially the same'' 
as if the cloned pebble were not there. More precisely, we show 

\begin{lemma} Suppose that \textsc{Explore2} is applied to the original instance (\emph{before} \textsc{Push}) and the 
new instance with the cloned pebble at $v^*$. Then at the end of each round $i$, $\act_i = \act^{\dagger}_i$ and $\C_i = \C^{\dagger}_i$, where 
$\act_i$, $\act^{\dagger}_i$ are the activated sets in the original and the new instances respectively, 
and $\C_i$ and $\C^{\dagger}_i$ are the conflict sets in the original and the new instances respectively. 
\label{lem:putClone}
\end{lemma}

\noindent \textbf{Step 2}: We then remove the original pebble at $u^*$ but keep the clone at $v^*$ (the same as 
the original instance after \textsc{Push}). 
Reapplying 
\textsc{Explore2}, we then show that in each round, the set of activated nodes and the conflict set 
cannot enlarge. To be precise, we show\footnote{Note that here we still 
refer to the instance with the cloned pebble at $v^*$ as the \emph{new} instance.}

\begin{lemma} Suppose that \textsc{Explore2} is applied to the 
new instance with the cloned pebble put at $v^*$ and 
the original instance (\emph{after} \textsc{Push}). Then at the end of each round $i$, 

\begin{enumerate}
\item 
$\bigcup_{\tau=0}^{i} \act'_\tau \subseteq \bigcup_{\tau=0}^{i} 
\act^{\dagger}_\tau$;
\item $\bigcup_{\tau=0}^{i} \C'_\tau \subseteq \bigcup_{\tau=0}^{i} \C^{\dagger}_\tau$; 
\item An edge not in $E(\bigcup_{\tau=0}^{i} \C^{\dagger}_\tau)$, if oriented in the original instance (after \textsc{Push}), 
must have the same orientation as in the new instance. 
\end{enumerate}

Here $\act^{\dagger}_i$, $\act'_i$ are the activated sets in the new and the original instance (after \textsc{Push}), respectively, and $\C^{\dagger}_i$ and $\C'_i$ are 
the conflict sets in the new and the original instances (after \textsc{Push}), respectively. 
\label{lem:removeOriginal}
\end{lemma}

Lemma~\ref{lem:putClone} and Lemma~\ref{lem:removeOriginal}(1) together imply Lemma~\ref{lem:generalPush} 
and we will prove the two lemmas in 
Sections~\ref{sec:putClone} and~\ref{sec:removeOriginal} respectively. 

The following lemma is convenient for proving Lemmas~\ref{lem:putClone} and~\ref{lem:removeOriginal} and we will prove it first. 
It states that 
the ``non-determinism'' in the order of fake orientations does not matter, 
allowing us to let the two instances ``mimic'' the behavior of each other when we compare the conflict sets in the main proofs.

\begin{lemma} In the sub-procedure \emph{Conflict set construction}, independent of the order of the edges being directed away from
the new 
conflict set $\C_i$, the final outcome is the same in the following sense. 

\begin{enumerate}
\item The sets of nodes in $\C_i$ is the same. 
\item Every edge not in $E(\C_i)$ has the same orientation. 
\end{enumerate}

\label{lem:nondeterminism}
\end{lemma}
\subsubsection{Proof of Lemma~\ref{lem:nondeterminism}}

\label{sec:nondeterminism}

We plan to break each system into a set of subsystems and use the following lemma recursively to prove the lemma. 

\begin{lemma} Let $T$ be a tree of neutral edges 
in the beginning of the sub-procedure \emph{Conflict set construction}
whose nodes are all in $V \backslash \bigcup_{\tau=0}^{i-1}\C_\tau$ and consist of only the following two types: 

\begin{enumerate}

\item Type $\alpha$: a node $v$ that (1) is already in $\C_i$ or has a directed path to a node in $\C_i$ in the beginning of the 
sub-procedure, or (2) at the end of all possible executions of the sub-procedure, it always has a directed path to some node 
in $\C_i \backslash T$. 

\item Type $\beta$: a node $v$ that (1) is not in $\C_i$ and does not have a directed path to a node in $\C_i$ in the beginning 
of the sub-procedure, and (2) at the end of all possible executions of the sub-procedure, it never has a directed path to some node in 
$\C_i$ via edges not in $T$. Furthermore, (3) all its incident neutral edges in the beginning of the sub-procedure are either in $T$, or never become directed towards $v$ in any execution. 

\end{enumerate}

Then the two properties of 
Lemma~\ref{lem:nondeterminism} hold. Namely, at the end of any execution, the final set 
$\C_i \cap T$ is the same and every edge in $T \backslash E(\C_i)$ has the same orientation. 

\label{lem:instrumental}
\end{lemma}

Intuitively, Type $\alpha$ nodes in $T$ are those bound to be part of $\C_i$ in any execution, while Type $\beta$ nodes 
may or may not become part of $\C_i$. If a Type $\beta$ node does become part of $\C_i$, 
then it must have a directed path to some Type $\alpha$ node in $T$ via the edges in $T$ after the execution. 
Notice also that by definition, a Type $\beta$ node cannot be overloaded (otherwise, it is part of $\act_0 \subseteq \C_0$).

\begin{proof} Let us first observe the outcome of an arbitrary execution of this sub-procedure. There can be 
two possibilities. 

\begin{itemize}
\item \textbf{Case 1}. The entire tree $T$ ends up being part of $\C_i$. 

\item \textbf{Case 2}. A set of sub-trees $T_1$, $T_2$, $\cdots$ become part of $\C_i$. The remaining nodes 
$T\backslash \bigcup_j T_j=\overline{F}$ form a forest. Each node $v \in \overline{F}$, if it has a non-$\overline{F}$ neighbor in $T$, 
then this neighbor is in some tree $T_j \subseteq \C_i$ and their shared 
edge is directed toward $v$. 

\end{itemize}

The following claim is easy to verify and useful for our proof. 

\begin{numberedClaim} Let $v \in T$ be a Type $\beta$ node, and suppose that $v$ has an incident edge in $T$ that becomes outgoing during the execution of the sub-procedure.
Then one of its incident edges in $T$ must become incoming first, and furthermore $ dl(v) + pl(v)+ rl^i(v) + \sum_{u: vu \in T} w_{vu} > (5/3 + \beta/3)t$, where $rl^i(v)$ is the rock load 
of $v$ in the beginning of the sub-procedure. 
\label{cla:notAll}
\end{numberedClaim}

We now consider the two cases separately. \\

\textbf{Case 1}: Suppose that in a different execution, the outcome is Case 2, i.e., there remains a forest $\overline{F} \subseteq T$ 
not being part of $\C_i$. 

Choose a tree $\overline{T}$ in $\overline{F}$ and then choose any node in $\overline{T}$ as the root $\overline{r}$. 
Define the level of a node in $\overline{T}$ as its distance 
to $\overline{r}$. Consider the set of nodes $v$ with the largest level $l$: they must be leaves of $\overline{T}$. 
By Proposition~\ref{pro:third}(3), in the new execution, 
all non-$\overline{F}$ neighbors of $v$ in $T$ direct their incident edges connecting $v$ towards $v$. As a result, 
by Claim~\ref{cla:notAll} and the fact that $v$ becomes part of $\C_i$ in the original execution, $v$ of level $l$ must direct its incident edge in $\overline{T}$ toward its neighbor of level $l-1$ in $\overline{T}$. 
Nodes of level $l-1$ then have incoming edges 
from their neighbors of level $l$ and from their non-$\overline{F}$ neighbors in $T$. So again they direct the edges in $\overline{T}$ 
towards the nodes of level $l-2$ in $\overline{T}$. Repeating this argument, we conclude that $\overline{r}$ receives all its incident edges in $T$ in the new execution, 
a contradiction to Claim~\ref{cla:notAll}. This proves Case 1. \\

\textbf{Case 2}: Let us divide the incident edges in $T$ of a node $v \in \overline{F}$ into three categories according to the outcome of the original
execution: incoming $E_{\inc}(v)$, outgoing $E_{\out}(v)$, and neutral $E_{\neu}(v)$. Notice that by Proposition~\ref{pro:third}(3), all edges connecting $v$ 
to its non-$\overline{F}$ neighbors in $T$ are in $E_{\inc}(v)$. Moreover, 
the following facts should be clear: at the end of 
any other execution, (1) an edge $e \in E_{\out}(v)$ must be directed 
away from $v$ if all edges in $E_{\inc}(v)$ are directed towards $v$, and (2) an edge in $E_{\inc}(v) \cup E_{\neu}(v)$ can be directed 
away from $v$ only if beforehand some edge in $E_{\out}(v) \cup E_{\neu}(v)$ is directed towards $v$, or $v$ ends up being part of $\C_i$. 

\begin{numberedClaim} Let $\overline{F} \subseteq T$ be the forest not becoming part of $\C_i$ in the original execution. In any other execution of the sub-procedure,

\begin{enumerate}
\item given $v \in \overline{F}$, it never happens that an edge $e \in E_{\out}(v) \cup E_{\neu}(v)$
is directed towards $v$ or an edge in $E_{\inc}(v)$ is directed away from $v$; 

\item none of the nodes in $\overline{F}$ ever becomes part of $\C_i$. 

\end{enumerate}
\label{cla:noReverse}
\end{numberedClaim}

\begin{proof} Suppose that (2) is false and $v \in \overline{F}$ is the first node becoming part of $\C_i$. Then some edge $e =v_0u \in E_{\inc}(v_0)$, 
where $v_0$ and $v$ are connected in $\overline{F}$ and $u\in T$ is a non-$\overline{F}$ neighbor of $v_0$, is directed towards $u$ beforehand. 
So (1) must be false first. Let $e'=v'u'$ be the first edge violating (1). (At this point, no node in $\overline{F}$ is part of $\C_i$ yet). 
If $e' \in E_{\out}(v') \cup E_{\neu}(v')$ is directed toward $v'$, then node $u'$ directs edge $e'$ towards $v'$ because it first has another edge $e'' \in E_{\out}(u') \cup E_{\neu}(u')$ 
coming toward itself. Then $e''$ should be the edge chosen, a contradiction. If $e' \in E_{\inc}(v')$ is directed away from $v'$, then some edge 
$e'' \in E_{\out}(v') \cup E_{\neu}(v')$ is directed toward $v'$ first, again implying that $e''$ should be chosen instead, another contradiction. Thus (1) and (2) hold. 
\qed\end{proof}

\begin{numberedClaim} Suppose that $T_j \subseteq \C_i$ in the original execution. Then in any other execution, 

\begin{enumerate} 

\item $T_j \subseteq \C_i$;

\item Every edge $e=vu$ with $v \in T_j$ and $u \in \overline{F}$ is directed toward $u$. 
\end{enumerate}
\label{cla:theSameAsCaseOne}
\end{numberedClaim}

\begin{proof}
For (1), we argue that $T_j$ itself satisfies the condition of Lemma~\ref{lem:instrumental} and is exactly Case~1. 
For this, we need to show that a Type $\beta$ node $v$ of $T$ in $T_j$ is also a Type $\beta$ node in $T_j$, i.e., 
$v$ never has a directed path to some node in $\C_i$ via edges not in $T_j$. As $v$ is a Type $\beta$ node in $T$, it suffices 
to show that it cannot have a directed path to some Type $\alpha$ node in $T \backslash T_j$ via edges in $T$. 
Suppose there is such a path $P$. Then $P$ must go through some node $u \in \overline{F}$, implying that $u$ becomes part of $\C_i$ in this execution, 
a contradiction to Claim~\ref{cla:noReverse}(2). 
This proves (1). (2) follows from Claim~\ref{cla:noReverse}(2) and Proposition~\ref{pro:third}(3). 
\qed\end{proof} 

What remains to be done is to show that 
all edges in $\overline{F}$ have the same orientation in any other execution. 
Let $L_0 \subseteq \overline{F}$ be the set of nodes $v$ satisfying $|E_{\inc}(v) \cap \overline{F}|=0$ and 
$L_{i>0} \subseteq \overline{F}$ be the set of nodes which can be reached from a node in $L_0$ by a directed path in $\overline{F}$ of maximum length exactly $i$ 
after the original execution. 
In any other execution, 
by Claim~\ref{cla:theSameAsCaseOne}(2), given $v \in L_0$, all edges in $E_{\inc}(v)$ are directed towards 
$v$, so all edges in $E_{\out}(v) \cap \overline{F}$ are directed away from $v$. Now an inductive argument on $i$, combined with Claim~\ref{cla:noReverse}(1), 
completes the proof 
of Case 2.

\qed\end{proof}

\begin{proof} (of Lemma~\ref{lem:nondeterminism}) We now explain how to make use of Lemma~\ref{lem:instrumental} to prove Lemma~\ref{lem:nondeterminism}. For this, we decompose each system into a set of subsystems that satisfy the conditions required in Lemma~\ref{lem:instrumental}.

First consider a system that is not a cycle. 
In the beginning of the sub-procedure \emph{Conflict set construction}, 
let $F$ be the forest consisting of the nodes in $V\backslash \bigcup_{\tau=0}^{i-1}\C_\tau$ and the edges that are neutral. 
We can assume that all nodes having a directed path to $\C_i$ are (already) in $\C_i$ as well. 

Create a graph $H$ whose node set are the connected components (trees) of $F$. If a non-$\C_i$ node in such a tree has a directed 
edge (we refer to the beginning of the sub-procedure) to some other non-$\C_i$ node in another 
tree, draw an arc from the node representing the former tree to the node representing the latter tree in $H$. 
(Intuitively, an arc in $H$ indicates the possibility that a node in the former tree becomes part of $\C_i$ 
because of a directed edge to a node in $\C_i$ in the latter tree). 
As the entire system is not a cycle, some node in $H$ must have out-degree 0. It is easy to verify that the particular tree corresponding to this node 
satisfies the conditions in Lemma~\ref{lem:instrumental}, so the lemma can be applied to it. 

We now find the next tree satisfying the conditions of Lemma~\ref{lem:instrumental} 
by redefining the graph $H$ as follows. 
Observe that the ``processed'' tree (the one we applied Lemma~\ref{lem:instrumental} to) has exactly two types of non-$\C_i$ nodes in the beginning of the sub-procedure: those that always become part of $\C_i$ (i.e., in every possible execution of the sub-procedure) and those that never become part of $\C_i$. Nodes in other trees that, in the beginning, have a directed edge to the former type of nodes are bound to become part of $\C_i$ (i.e., they satisfy the conditions of a Type $\alpha$ node in their tree). Nodes in other trees with a directed edge to the latter type of nodes are not to become part of $\C_i$ because of them. So in 
$H$, we can just remove the corresponding arcs and the node representing the already processed tree. 
In the updated $H$, the node with out-degree 0 is the next tree, on which Lemma~\ref{lem:instrumental} can be applied. 
Repeating this procedure, we are done with the first case (when the system is not a cycle). 

Finally, consider the case that the entire system is a cycle. For the special case that the entire cycle consists of neutral edges, it is easy to verify 
that Lemma~\ref{lem:nondeterminism} holds. So suppose that the set of neutral edges form a forest (precisely, a set of disjoint paths). 
We can proceed as before---build $H$ and find a vertex in $H$ with out-degree 0 and recurse---except for 
the special case that $H$ is a directed cycle $V_1$, $V_2$,$\dots$ in the beginning. 
Observe that 
the last node $v \in V_1$ has a directed edge to the first node $u \in V_{2}$ and neither $v$ nor $u$ is in $\C_i$. Similarly, the last node of $V_{2}$ is also not in $\C_i$ and 
neither is the first node of $V_3$ and so on. In this case, it is easy to see that 
Lemma~\ref{lem:nondeterminism} holds for the entire system. 
\qed\end{proof}

\subsubsection{Proof of Lemma~\ref{lem:putClone}} 
\label{sec:putClone}

When \textsc{Explore2} is applied on the original instance before \textsc{Push}, suppose that $v^*$ joins the conflict set 
in round $k$, i.e., $v^* \in \C_k$. We first make the following claim. 

\begin{numberedClaim} Apply \textsc{Explore2} to the new instance. In round $k$, immediately after the sub-procedure 
\emph{Conflict set construction}, the following holds.

\begin{enumerate}
\item $\act_\tau = \act^{\dagger}_{\tau}$, for $0 \leq \tau \leq k$, 
\item $\C_\tau = \C^{\dagger}_{\tau}$, for $0 \leq \tau \leq k$, 
\item Edges not in $E( \bigcup_{\tau=0}^{k}\C_\tau)$ have the same orientations 
in both instances. 
\end{enumerate}

\label{cla:roundK}
\end{numberedClaim} 

We will prove the claim shortly after. 
In the following, we will show that $\act_{k} = \act^{\dagger}_{k}$ at the end of round $k$. Combining this with Claim~\ref{cla:roundK}(2)(3) and Lemma~\ref{lem:nondeterminism}, an inductive argument
proves that Lemma~\ref{lem:putClone} is true also from round $k$ onwards. 

Recall that by the definition of \textsc{Push}, at the end of \textsc{Explore2} in the original instance, either (1) $\child(v^*)=\emptyset$, 
or (2) $dl(v^*) + pl(v^*) + w_{v^*u} \leq (5/3-2/3\cdot\beta)t$ for all $u \in \father(v^*)$. We consider these two cases separately. 

\textbf{Case 1}: Suppose that $\child(v^*)=\emptyset$ in the original instance at the end of \textsc{Explore2}. We will show that 
at the end of round $k$, $\act_k = \act^{\dagger}_k$ and in particular $v^* \in \act_k = \act^{\dagger}_k$. 
By Claim~\ref{cla:roundK}(1), we just have to argue that a node is activated by Rule 1 or Rule 2 in the original instance if and only if it is activated by one of these two rules in the new instance, in round $k$. 

For $v^*$, recall that it is part of $\C_k$. It becomes so by either (1) being a Type A node in $\act_k$, or (2) having an outgoing edge $v^*u$ 
and $u \in \bigcup_{\tau=0}^{k}\C_{\tau}$. For the former case, Claim~\ref{cla:roundK}(1) shows that $v^* \in \act^{\dagger}_k$. For the latter case, 
as $\child(v^*)=\emptyset$ at the end of \textsc{Explore2} in the original instance, in round $k$, $dl(v^*)+ pl(v^*)+ w_{v^*u} > (5/3+\beta/3)t$, 
and hence Rule 1 applies to $v^*$. In the new instance, the preceding inequality still holds since the pebble load of $v^*$ is increased by the cloned pebble. As 
$u \in \bigcup_{\tau=0}^{k}\C^{\dagger}_{\tau}$ (Claim~\ref{cla:roundK}(2)), Rule 1 again applies to $v^*$ (note that $u$ is still a father of $v^*$, since otherwise $v^*$ would be overloaded and part of both $\act^{\dagger}_0$ and $\act_0$). 

For other nodes $v \neq v^*$, as $pl(v)+dl(v)$ are the same in both instances, if $v$ is activated by Rule 1 in the original instance, then it is so too in the new instance, and vice versa. 
We have established that the set of nodes activated by Rule 1 is the same in both instances. Now by Claim~\ref{cla:roundK}(2), 
the set of nodes activated by Rule 2 is again the same in both instances. 
Therefore, $\act_k = \act^{\dagger}_k$ at the end of round $k$.

\textbf{Case 2}: Suppose that $dl(v^*) + pl(v^*) + w_{v^*u} \leq (5/3-2/3\cdot\beta)t$ for all $u \in \father(v^*)$ 
in the original instance. Then $v^*$ cannot be a Type B node in the original instance, i.e., it is not activated by Rule 1 (but it is possible that 
$v^*$ is activated by Rule 2 or as a Type A node). 
We now argue that 
in the new instance, in round $k$, $v^*$ cannot be activated by Rule 1 either. 

By the definition of \textsc{Push} (specifically Definition~\ref{def:complicatedPush}(3)(4)), in the original 
instance, each father and child $u \in \bigcup_{\tau=0}^{k}\C_{\tau}$ of $v^*$ satisfies $dl(v^*) + pl(v^*) + w_{v^*u} \leq (5/3-2/3\beta)t$ (notice that when we compare original and new instance, a father can become a child and vice versa). Therefore, even with the cloned pebble (of weight at most $\beta t$) in the new instance, Rule 1 still cannot be applied to $v^*$ in round $k$. 

For other nodes $v \neq v^*$, it is easy to see that $v$ is activated by Rule 1 in the original instance if and only if in the new instance in round $k$. 
We have established that the set of nodes activated by Rule 1 is the same in both instances. Now by Claim~\ref{cla:roundK}(2), 
the set of nodes activated by Rule 2 is again the same in both instances. Therefore, $\act_k = \act^{\dagger}_k$ at the end of round $k$. \\

\emph{Proof of Claim~\ref{cla:roundK}}: 
Consider the moment at the end of round $k-1$ when \textsc{Explore2} is applied on the original instance before \textsc{Push}. 
In the special case of $k=0$, we refer to the moment immediately after \textsc{Forced Orientations} 
is called in the initialization of \textsc{Explore2}. 

In this moment, let us put the cloned pebble at $v^*$ 
and invoke \textsc{Forced Orientations}. This causes a (possibly empty) set of neutral edges $\overline{E}$ to become directed. 
Let $V_0$ be the set of nodes which are the heads or tails of the now directed edges in $\overline{E}$. 
Let $V_1$ be the set of nodes that can be arrived at from nodes in $V_0$ following the other directed edges $E^*$ 
(i.e., those that are already oriented at the end of round $k-1$ before the cloned pebble is put at $v^*$). 
Observe that $v^* \in V_0$ can 
reach any node in $V_0 \cup V_1$ by following the directed edges in $\overline{E} \cup E^*$. 
Let $E_{\inc}(v)$, $E_{\out}(v)$, and $E_{\neu}(v)$ denote the set of incident incoming, outgoing, neutral 
edges of each node $v \in V$ after we put the cloned pebble and called \textsc{Forced Orientations}. It should be clear that 
(1) $\overline{E} \subseteq \bigcup_{v \in V_0} E_{\out}(v)$, (2) $ \bigcup_{v \in V_0\cup V_1} 
E_{\out}(v) \cap \delta(V_0 \cup V_1) = \emptyset$, and (3) none of the nodes in $V_0$ is overloaded at the end of round $k-1$ (and hence also not in subsequent rounds).

\begin{numberedClaim} When \textsc{Explore2} is applied on the original instance before \textsc{Push}, 

\begin{enumerate}

\item If an edge $e$ is in $\overline{E} \cap E_{\out}(v)$ for some $v \in V_0$, then 
at the end of round $k$, edge $e$ is also an outgoing edge of $v$ (independent of the order of fake orientations);

\item At the end of round $k-1$, none of the nodes in $V_0 \cup V_1$ is part of the conflict set built so far, i.e. 
$(V_0 \cup V_1) \cap \bigcup_{\tau=0}^{k-1} \C_{\tau} = \emptyset$.

\end{enumerate}
\label{cla:explore2First}
\end{numberedClaim}

\begin{proof} Consider the edge $v^*u \in \overline{E} \cap E_{\out}(v^*)$. As $v^*$ is part of $\C_k$, at the end of round $k$, $v^*u$ cannot 
be neutral. As it is directed toward $u$ after the added cloned pebble, 

\begin{equation} 
w_{v^*u} + pl(v^*) + dl(v^*) + rl^{k-1}(v^*) + w > (5/3+ \beta/3)t, 
\label{equ:tooMuch}
\end{equation}

\noindent where $w$ is the weight of the cloned pebble and $rl^{k-1}(v^*) $ is the weight of the rocks assigned to $v^*$ 
at the end of round $k-1$. Suppose for a contradiction that edge $v^*u$ is directed toward $v^*$ at the end of round $k$. 
Recall that by Definition~\ref{def:complicatedPush}(3), for the pebble to be 
pushed from $u^*$ to $v^*$ in the original instance, $pl(v^*) + dl(v^*) + rl(v^*) \leq (5/3-2/3\cdot \beta)t$, where $rl(v^*)$ 
is the weight of the rocks assigned to $v^*$ at the end of \textsc{Explore2}. Then 

$$ dl(v^*) + pl(v^*)+ (rl^{k-1}(v^*)+ w_{v^*u})+w \leq dl(v^*) + pl(v^*)+ rl(v^*)+w \leq (5/3+\beta/3)t,$$ 

\noindent a contradiction to inequality~(\ref{equ:tooMuch}). So we establish that $v^*u$ is directed toward $u$ at the end of round $k$. 
Consider $u$ and its incident edge $uu' \in \overline{E} \cap E_{\out}(u)$. 
The fact that $v^*u$ causes $uu'$ to be directed toward $u'$ 
implies that at the end of round $k$, 
$uu'$ cannot be directed toward $u$ or stay neutral. Repeating this argument, we prove (1). 

If a node in $V_0\cup V_1$ is part of $\bigcup_{\tau=0}^{k-1}\C_\tau$, then either $v^*$ is part of $\bigcup_{\tau=0}^{k-1}\C_\tau$, 
a contradiction to the assumption that $v^*$ joins the conflict set in round $k$, or some node in $V_0 \backslash \{ v^* \}$ has an incident edge in $\overline{E}$ directed away from it at the end of round $k-1$ (see Proposition~\ref{pro:third}(3)), a contradiction to the definition of $\overline{E}$. This proves (2). 


\qed\end{proof}

Claim~\ref{cla:explore2First}(2) has the important implication that, in the original instance, the set of nodes $V_0 \cup V_1$ is ``isolated'' from the rest of the graph 
up to the end of round $k-1$ in \textsc{Explore2}: they do not have a directed path to nodes in $\bigcup_{\tau=0}^{k-1}\C_{\tau}$ and they are not reachable from nodes 
in $\bigcup_{\tau=0}^{k-2}\act_{\tau}$. 

\begin{numberedClaim} Suppose that $k \geq 1$. 
When \textsc{Explore2} is applied on the new instance, at the end of round $k-1$, 
\begin{enumerate}
\item Every edge $e \in E_{\out}(v)$ (respectively $E_{\inc}(v)$, $E_{\neu}(v)$) for any $v \in V_0 \cup V_1$ is an outgoing (respectively incoming, neutral) edge of $v$ in the new instance; 
\item $\act_{\tau} = \act^{\dagger}_{\tau}$ for $0 \leq \tau \leq k-1$;
\item $\C_{\tau} = \C^{\dagger}_{\tau}$, for $0 \leq \tau \leq k-1$; 
\item Every edge not in $E(\bigcup_{\tau=0}^{k-1} \C_{\tau}) \cup \overline{E}$ has the same orientation in both instances. 
\end{enumerate}

\label{cla:roundKminusone}
\end{numberedClaim}

\begin{proof} By Claim~\ref{cla:explore2First}(2), none of the nodes in $V_0 \cup V_1$ is overloaded in the original instance, as $k\geq 1$.  
By Lemma~\ref{lem:nondeterminism}, we may assume that both instances decide their fake orientations based on the same fixed total order. 
Let us define the following events for both instances: 
\begin{itemize}
\item $\beta$: An edge in $E(V_0 \cup V_1)$ becomes directed.
\item $\alpha_1$: An edge not in $E(V_0 \cup V_1)$ becomes directed.
\item $\alpha_2$: The sub-procedure \emph{Activation of nodes} is executed.
\item $\alpha_3$: A new round starts and a set of (Type A-) nodes is activated.
\item $\alpha_4$: The internal while-loop of \emph{Conflict set construction} is executed and a set of nodes is added into the conflict set.
\end{itemize}
Using an inductive argument, the following fact can easily be verified: 
\begin{quote}
As long as no edge in $E_{\neu}(v) \cup E_{\inc}(v)$ becomes outgoing for any $v \in V_0 \cup V_1$, the sequences of $\alpha$-events are the same in both instances (but possibly intermitted by different sequences of $\beta$-events) up to the end of round $k-1$. Furthermore, right after two corresponding $\alpha$-events in the original and new instance, the conflict set and activated nodes, and the direction of all edges not in $E(V_0 \cup V_1)$ are the same in both instances.
\end{quote}
To prove (1), consider the first moment in the new instance when an edge $e \in E_{\neu}(v) \cup E_{\inc}(v)$ becomes outgoing for any $v \in V_0 \cup V_1$ before the end of round $k-1$.
For this to happen, as $v$ is not overloaded in the original instance, 
there exists another edge $e' = vu \in E_{\neu}(v) \cup E_{\out}(v)$ becoming incoming for $v$ first. 
By the above fact, it must be the case that $u \in V_0 \cup V_1$. 
Then $e' \in E_{\neu}(u) \cup E_{\inc}(u)$, and $e'$ becomes outgoing for $u$ before $e$ becomes outgoing for $v$, a contradiction. 

We next show that every edge $e \in E_{\out}(v)$ is an outgoing edge for $v \in V_0 \cup V_1$ at the end of round $k-1$ in the new instance. Assume that $v^*$'s system is not a cycle. Then $E_{\inc}(v^*) \subseteq \delta(V_0 \cup V_1)$, and all these edges are incoming at the end of round $k-1$, implying that all edges in $E_{\out}(v^*)$ must be outgoing. Now an inductive 
argument on the rest of the nodes $v \in V_0 \cup V_1$ (based on their distance to $v^*$) establishes that 
$e \in E_{\out}(v)$/$E_{\inc}(v)$/$E_{\neu}(v)$ 
is an outgoing/incoming/neutral edge of $v$ at the end of round $k-1$ in the new instance. The cycle-case follows by a similar argument. 
This completes the proof of (1). 

Finally, combining (1) with the above fact, the rest of the claim follows. 
\qed\end{proof}

\begin{numberedClaim} Suppose that $k=0$. When \textsc{Explore2} is applied on the new instance, at the end of the initialization (after \textsc{Forced Orientations}),  

\begin{enumerate}
\item Every edge $e \in E_{\out}(v)$ (respectively $E_{\inc}(v)$, $E_{\neu}(v)$) for any $v \in V_0 \cup V_1$ is an outgoing (respectively incoming, neutral) edge of $v$ in the new instance; 
\item Every edge not in $E(V_0 \cup V_1)$ has the same orientation in both instances;
\item The set of overloaded nodes are the same in both instances. 
\end{enumerate}

\label{cla:round0Minusone}
\end{numberedClaim}

\begin{proof} In the new instance, we claim that no edge in $E_{\neu}(v) \cup E_{\inc}(v)$ becomes outgoing for any $v \in V$ during the initialization. Suppose not and 
$e \in E_{\neu}(v) \cup E_{\inc}(v)$ is the first such edge. If this happens because another edge $e' =vu 
\in E_{\out}(v) \cup E_{\neu}(u)$ is directed toward $v$ first, then $e'$ should have been chosen. 
So $v$ must be overloaded in the original instance and by Lemma~\ref{lem:forcedOrientationLemma}, $e$ is the only edge in 
$E_{\inc}(v)$ and $dl(v) + pl(v) + w_{e=vu_0} > (5/3+ \beta/3)t$. (Notice that  $v \neq v^*$). 

Consider the moment in the initialization of the original instance, when $e = vu_0$ is directed toward $v$.
First suppose that in this moment, $u_0$ has no incoming edges yet. Then we know that $dl(u_0) + pl(u_0) + w_{vu_0} > (5/3+ \beta/3)t$ and the pair $(u_0,v)$ precedes $(v,u_0)$ in the total order of edges. This is still true in the new instance, contradicting our assumption that $vu_0$ is chosen to be directed toward $u_0$.
So $u_0$ already has some incoming edges $E_{u_0}$ in the original instance. In the new instance, when $vu_0$ is directed toward $u_0$, it cannot be that 
all edges of $E_{u_0}$ are already directed toward $u_0$. 
So at least one such $u_1u_0 \in E_{u_0}$ is still neutral (it cannot be outgoing 
because of the choice of $e=vu_0$). 
Repeating this argument, in the new instance, 
we find a path of neutral edges $vu_0u_1\dots$ immediately before $e=vu_0$ is directed toward $u_0$, 
and this path ends at a node $u_z$ where $dl(u_z) + pl(u_z) + w_{u_zu_{z-1}} > (5/3+\beta/3)t$, and the pair 
$(u_z,u_{z-1})$ precedes the pair $(v,u_0)$. This contradicts the assumption that $e=vu_0$ is chosen to be directed toward $u_0$. 

So we established that no edge in $E_{\neu}(v) \cup E_{\inc}(v)$ becomes outgoing for any $v \in V$. To complete the proof of (1) and (2), 
suppose that $uv \in E_{\inc}(v)$ for some $v \in V$ remains neutral after the initialization of the new instance. Then there must be another edge $wu \in E_{\inc}(u)$ which also remains neutral. Repeating this argument, we conclude that the entire system is a cycle, whose edges are all neutral after the initialization of the new instance. As $uv \in E_{\inc}(v)$, there must be some edge $xy$ in this cycle, 
so that $dl(x) + pl(x) + w_{xy} > (5/3 + \beta/3)t$ after the cloned pebble is put on $v^*$. This edge cannot remain neutral after the initialization of the new instance, a contradiction.

Finally, (3) follows from (1) and (2), and the fact that no node in $V_0$ is overloaded in both instances. 

\qed\end{proof}

To complete the proof of Claim~\ref{cla:roundK}, we now show that in round $k$, after the sub-procedure \emph{Conflict set construction}, the outcome of the
two instances are exactly the same, except for the orientation of the edges in $E(\bigcup_{\tau=0}^{k} \C_{\tau})$. 
Notice that by Claim~\ref{cla:roundKminusone}(1)(4) and Claim~\ref{cla:round0Minusone}(1)(2), at the end of round $k-1$,
the orientations of all edges not in $E(\bigcup_{\tau=0}^{k-1} \C_{\tau})$ are the same in both instances, with the only exception that $\overline{E}$ are oriented in the new instance but neutral
in the original instance. Furthermore, by Claim~\ref{cla:roundKminusone}(2) and Claim~\ref{cla:round0Minusone}(3), 
the same set of nodes are added into $\act^{\dagger}_k$, 
$\act_k$, $\C^{\dagger}_k$, $\C_k$ in the beginning of round $k$ (as Type $A$ nodes). 

Let $V'_1 \subseteq V_1$ be the set of nodes that can be reached by a directed path from $v^*$ in the original instance at the end of round $k-1$ (such a path does not use edges in $\overline{E}$).
Let us first suppose the system containing $v^*$ is a tree.
In the following, when we say the ``sub-tree'' of an edge $e \in E_{\neu}(v)$ for some $v\in V_0 \cup V_1$, we mean
the sub-tree outside of $V_0 \cup V_1$ connected to $V_0 \cup V_1$ by the edge $e$ (note that $e \in \delta(V_0 \cup V_1)$).
We now make use of Lemma~\ref{lem:nondeterminism} to let the two instances mirror each other's behavior.
Consider how $v^*$ becomes part of $\C_k$ in the original instance.

\textbf{Case 1}: in the beginning of round $k$, $v^*$ or some node in $V'_1$ becomes a Type $A$ node.
Then $v^*$ becomes part of the conflict set in both instances in the beginning of the sub-procedure \emph{Conflict set construction}, before any further edges become directed.
In this case, in the original instance, let $v^*$
direct all edges in $\overline{E}$ away from $v^*$ (by running ahead a few iterations and picking the respective edges incident with $v^*$ as fake orientations). After that, in both instances, direct all remaining neutral edges incident with $v^*$ away from $v^*$. Now the two instances are the same\footnote{When we say that two instances are \emph{the same} at a certain time point, we mean that the conflict set and activated nodes, and the orientation of all edges not in $E(\bigcup_{\tau=0}^{k-1} \C_{\tau})$ are the same.}, and we can let them continue identically until the end of the sub-procedure (also note that all edges incident with $v^*$ are already oriented in both instances).

\textbf{Case 2}: in the sub-procedure \emph{Conflict set construction}, due to fake orientations in the sub-trees of edges in $\cup_{v \in V'_1}E_{\neu}(v)$,
some nodes in $V'_1$ (hence $v^*$) become part of $\C_k$. In this case, in both instances, apply these fake orientations first.
Then $v^*$ becomes part of the conflict set in both instances. Let the original instance direct the edges in $\overline{E}$ away from $v^*$, and then, in both instances, direct all remaining neutral edges incident with $v^*$ away from $v^*$. Now the
two instances are the same, and we can let them continue identically until the end of the sub-procedure.

\textbf{Case 3}: the above two cases do not apply. Consider the execution of the sub-procedure \emph{Conflict set construction} in the original
instance in round $k$. $E_{\neu}(v^*)$ can be partitioned into $E_{\neu\rightarrow \inc}(v^*)$ and
$E_{\neu\rightarrow \out}(v^*)$, the former (latter) being those edges in $E_{\neu}(v^*)$ becoming incoming (outgoing) inside the sub-procedure.

Observe that (1) $E_{\neu\rightarrow \inc}(v^*)\neq \emptyset$, otherwise $v^*$ cannot become part of $\C_k$ 
in the original instance (see Claim~\ref{cla:explore2First}(1)), 
and (2) in round $k$, as long as no edge $E_{\neu \rightarrow \out}(v^*)$ is directed toward $v^*$, then even with the cloned pebble at $v^*$,
 a proper subset $E' \subset E_{\neu\rightarrow \inc}(v^*)$ directed toward $v^*$ cannot cause another edge
in $E_{\neu\rightarrow \inc}(v^*) \backslash E'$ to be directed away from $v^*$ (by Definition~\ref{def:complicatedPush}.3
and the fact that $rl(v^*)=\sum_{e \in E_{\neu \rightarrow \inc}(v^*) \cup E_{\inc}(v^*)}
w_{e}$ in the original instance after round $k$).

Let the original instance start round $k$ with the fake orientations in the sub-trees of edges in $E_{\neu \rightarrow \inc}(v^*)$ until all edges in $E_{\neu \rightarrow \inc}(v^*)$ are directed toward $v^*$, and let the new instance mimic. Now the edges in $\overline{E}$ are directed away from $v^*$ also in the original instance (since any rock edge is heavier than the cloned pebble). Hence, all edges not in $E(\bigcup_{\tau=0}^{k-1} \C_{\tau})$ have the same orientations in both instances,
except that possibly some edges in $E_{\neu\rightarrow \out}(v)$ and in their sub-trees are already oriented in the new instance while not in the original
instance (this is because in the new instance, the pebble load at $v^*$ is higher). 
Let $E_{\neu\rightarrow \out}' \subseteq E_{\neu\rightarrow \out}(v^*)$ be those edges in $E_{\neu\rightarrow \out}(v^*)$ that are already oriented in the original instance at this point.
Now let the original instance apply all possible fake orientations in the sub-trees of edges
in $\bigcup_{v \in V_0\cup V_1 \backslash \{v^*\}} E_{\neu}(v) \cup E_{\neu\rightarrow \out}'$ and let the new
instance mimic. After this step, $v^*$ must be part of the conflict set $\C_k$ and $\C^{\dagger}_k$ in both instances. 
Finally, in both instances, direct all remaining neutral edges in $E_{\neu\rightarrow \out}(v^*)$ away from $v^*$. Now the
two instances are the same, and we can let them continue identically until the end of the sub-procedure. This finishes the proof of the tree case.

Next suppose that the system containing $v^*$ is a cycle. In the original instance, $v^*$ joins $\C_k$ in two possible ways.
Either $v^*$ or some node in $V'_1$ is a Type $A$ node (then this is the same as Case 1 above), 
or during the sub-procedure \emph{Conflict set construction} an edge $e_\alpha$ is directed toward $v^*$, causing the other incident edge $e_\beta$ to be directed away from $v^*$.
In this case, by Definition~\ref{def:complicatedPush}.4, $\overline{E} = \emptyset$. Hence, the two instances are the same already in the beginning of the sub-procedure, and we can let them perform identically by choosing the same fake orientations. This finishes the cycle case and the entire proof of Claim~\ref{cla:roundK}. 

\subsubsection{Proof of Lemma~\ref{lem:removeOriginal}} 
\label{sec:removeOriginal}

Our idea is to make use of Lemma~\ref{lem:nondeterminism}: we will apply \textsc{Explore2} simultaneously 
to both instances and let the new instance mimic the behavior of the original instance. In the following, we implicitly assume that 
nodes having a directed path to $\bigcup_{\tau=0}^{i}\C^{\dagger}_\tau$ (respectively 
$\bigcup_{\tau=0}^{i}\C'_\tau$) are part of it in the new (original) instance. 
Furthermore, at any time point considered, we refer to the current content of the sets $C^{\dagger}_i$ and $\C'_i$. 
The lemma below explains how the mimicking is done. 

\begin{lemma} In round $i \geq 0$, suppose that both instances are in the sub-procedure 
\emph{Conflict set construction} and Lemma~\ref{lem:removeOriginal}(2)(3) hold. Let the original instance apply an arbitrary fake orientation and invoke \textsc{Forced Orientations}. Then the new instance can apply a number of 
fake orientations so that Lemma~\ref{lem:removeOriginal}(2)(3) still hold. 

\label{lem:imitationGame}
\end{lemma}

\begin{proof}

In the original instance, suppose that the chosen fake orientation is to direct the edge $e_0= v_0u_0$ toward $u_0$. In the subsequent call of \textsc{Forced Orientations}, a tree $T_{u_0}$ of neutral edges are further directed away from $u_0$. Given two incident 
edges $e$, $e'$ of a node $v \in T_{u_0}$, we write $e \prec e'$ if $e$ is closer to $u_0$ than $e'$. Similarly, given two adjacent nodes $v, u \in T_{u_0}$, we write $v \prec u$ if $v$ is closer to $u_0$ than $u$. We make 
an important observation. 

\begin{numberedClaim} Suppose that $v \in T_{u_0}$ and $v \not \in \bigcup_{\tau=0}^i \C^{\dagger}_\tau$. Furthermore, suppose that $e, e' \in T_{u_0}$
are incident on $v$ and $e \prec e'$. Then $v$ can take at most one of them in the new instance, i.e.,
if either of them is directed toward $v$, then (after \textsc{Forced Orientations}) the other must be directed away from $v$.

In the special case of $v = u_0 \not \in \bigcup_{\tau=0}^i \C^{\dagger}_\tau$, assuming that $e$ is an incident edge of $u_0$
in $T_{u_0}$, $u_0$ can take at most one of $e_0=v_0u_0$ and $e$.

\label{cla:atMostOne}
\end{numberedClaim}

\begin{proof} The dedicated load $dl(v)$ and the pebble load $pl(v)$ are at least as heavy in the new instance as in the original. 
An edge not in $E(\bigcup_{\tau=0}^{i}\C^{\dagger}_{\tau})$, if oriented in the original instance, 
must be oriented in the same way in the new instance. So the rock load $rl(v)$ is also at least as heavy in the new 
instance as in the original. Thus, if in the original instance, $e$ being directed toward $v$ causes $e'$ to be directed 
away from $v$, then $v$ can take at most one of them in the original, and hence in the new instance. 

The second part of the claim follows from the same reasoning. 

\qed\end{proof}

Our goal is to apply a number of fake orientations in the new instance, so that the edges
$(\{e_0\}\cup T_{u_0}) \backslash E(\bigcup_{\tau=0}^{i} \C^{\dagger}_\tau)$ are directed the same way as in the original 
instance. 

First, if $e_0$ is still neutral in the new instance, direct it toward $u_0$ and invoke \textsc{Forced Orientations}. Notice that if $e_0$ is already directed toward $v_0$ in the new instance, then both $v_0, u_0 \in \bigcup_{\tau=0}^i \C^{\dagger}_\tau$, and hence $e_0 \in E(\bigcup_{\tau=0}^{i} \C^{\dagger}_\tau)$.

We make another observation. 

\begin{numberedClaim} In the new instance, after a call of \textsc{Forced Orientations}, assume that $e =vu \in T_{u_0}$, and $v \prec u$.

\begin{enumerate}

\item If $e=vu$ is directed toward $v$, then both $u, v \in \bigcup_{\tau=0}^i \C^{\dagger}_\tau$. 

\item If $e=vu$ is directed toward $u \not \in \bigcup_{\tau=0}^{i} \C^{\dagger}_{\tau}$, then the entire sub-tree of $T_{u_0}$ rooted at $u$ is directed away from $u_0$ and none of its nodes is in 
$\bigcup_{\tau=0}^{i} \C^{\dagger}_{\tau}$.

\end{enumerate}

\label{cla:structureTU0}
\end{numberedClaim}

\begin{proof} For (1), suppose that $e$ is directed toward $v$. 
If $v \in \bigcup_{\tau=0}^{i} \C^{\dagger}_{\tau}$, then so is $u$ and the claim holds. So assume that 
$v \not \in \bigcup_{\tau=0}^{i} \C^{\dagger}_{\tau}$. Consider the incident edge $e' \in T_{u_0}$ of $v$ 
with $e' \prec e$. By Claim~\ref{cla:atMostOne}, $e'$ must also be directed toward $u_0$. Repeating this argument, 
we find a sequence of edges directed toward $u_0$ and they either end up at a node in 
$\bigcup_{\tau=0}^{i} \C^{\dagger}_{\tau}$ (then implying that $v$ is part of 
$\bigcup_{\tau=0}^{i} \C^{\dagger}_{\tau}$, a contradiction), or at $u_0$ and $u_0 \not \in 
\bigcup_{\tau=0}^{i} \C^{\dagger}_\tau$. Then, by Claim~\ref{cla:atMostOne}, the edge $v_0u_0$ must be directed toward $v_0$ in the new 
instance, again implying that $v$ is part of 
$\bigcup_{\tau=0}^{i} \C^{\dagger}_{\tau}$, a contradiction. 

(2) is the consequence of Claim~\ref{cla:atMostOne} and our assumption that 
all nodes having a directed path to $\bigcup_{\tau=0}^{i} \C^{\dagger}_{\tau}$ are part of it. 
\qed\end{proof}

In the new instance, the set of neutral edges in $T_{u_0}$ form 
a set of node-disjoint trees $T_1, T_2, \dots$, where each tree $T_j$ has a root node $r_j$ that is closest 
to $u_0$ in $T_{u_0}$ ($r_j$ could be $u_0$ itself). 
Observe that no node in $T_j$ can be part of $\bigcup_{\tau=0}^{i-1}\C^{\dagger}_{\tau}$, since otherwise its incident edges would not be neutral in round $i$. 
It follows from Claim~\ref{cla:structureTU0} (resp. the last part of Claim~\ref{cla:atMostOne} if $r_j = u_0$) that $r_j \in \C^{\dagger}_i$. 
Hence, we can let the new instance direct 
the neutral edges in $T_j$ incident on $r_j$ 
away from it. 
If some edge in $T_j$
remains neutral after this, by Claim~\ref{cla:atMostOne}, there must exist a node $v \in T_j \cap \C^{\dagger}_i$ with neutral incident edges in $T_j$. Then again let $v$ direct all remaining neutral edges in $T_j$ 
away from it and continue this process until all edges in $T_j$ are directed away from $u_0$. 

By the above mimicking, we guarantee that all edges in $(\{e_0\}\cup T_{u_0}) \backslash E(\bigcup_{\tau=0}^{i} \C^{\dagger}_\tau)$ 
are directed the same way in both instances. This implies that Lemma~\ref{lem:removeOriginal}(3) holds after the mimicking. Next we argue that if a node $v$ is added into $\C'_i$ in the original instance, then it is either already in $\bigcup_{\tau=0}^{i}\C^{\dagger}_{\tau}$, or is added into 
$\C^{\dagger}_i$ as well after the mimicking. For $v$ to be added into $\C'_i$ in the original instance, it must have a directed path $P$ to some node $\hat{v} \in \C'_i$ 
after $v_0u_0$ is oriented toward $u_0$, where $\hat{v}$ is part of $\C'_i$ already before the fake orientation. Note that $\hat{v}$ is also in $\bigcup_{\tau=0}^{i}\C^{\dagger}_{\tau}$ before the mimicking. 
Let $\overline{v}$ be the first node on $P$ (starting from $v$) that is part of $\bigcup_{\tau=0}^{i}\C^{\dagger}_{\tau}$ after the mimicking. If $\overline{v}=v$, we are done. Otherwise, since Lemma~\ref{lem:removeOriginal}(3) holds after the mimicking, $v\not \in \bigcup_{\tau=0}^{i}\C^{\dagger}_{\tau}$ has a directed path to $\overline{v}\in \bigcup_{\tau=0}^{i}\C^{\dagger}_{\tau}$ in the new instance, a contradiction.  

So we have established Lemma~\ref{lem:removeOriginal} (2) and (3) after the mimicking. 
\qed\end{proof}

We use the above lemma to prove Lemma~\ref{lem:removeOriginal} for the case of $i \geq 1$. 

\begin{lemma} Suppose that Lemma~\ref{lem:removeOriginal} holds at the end of round $i-1$ for $i\geq 1$. 
Then it holds still at the end of round $i$. 
\label{lem:roundMoreThanZero}
\end{lemma}

\begin{proof} In round $i$, it is easy to verify that the Lemma~\ref{lem:removeOriginal} is true in the beginning of the sub-procedure
\emph{Conflict set construction}. Now let the original instance apply all the fake orientations and let the new instance mimic, using 
Lemma~\ref{lem:imitationGame}. Next let the new instance finish off its fake orientations arbitrarily. It is easy to see that 
Lemma~\ref{lem:removeOriginal} holds at the end of round $i$. 

\qed\end{proof}

We now handle the more difficult case of round $0$. Unlike the later rounds, Lemma~\ref{lem:removeOriginal} does not hold 
in the beginning of the sub-procedure \emph{Conflict set construction}: the set of overloaded nodes can be different 
in the two instances and the conflict set in the new instance may not be a superset of the conflict set in the original instance. 

In the following, we postpone the fake orientations of the original instance and just let the new instance perform some fake 
orientations until Lemma~\ref{lem:removeOriginal}(2)(3) hold. 

\begin{lemma} Consider the beginning of the sub-procedure 
\emph{Conflict set construction} in round $0$. In the new instance, as long as an edge 
$e=vu \in E(\C'_0)$ remains neutral and $v$ is part of $\C^{\dagger}_0$, direct $e$ toward $u$. Then finally, 
$\C'_0 \subseteq C^{\dagger}_0$. 
\label{lem:finallyTreesBigger}
\end{lemma}

\begin{proof}\footnote{The proof here is very similar to the proof of Lemma~\ref{lem:instrumental}, Case 1. So we 
only sketch the ideas.} Consider a connected 
component $H$ in the induced subgraph $G_\rock[\C'_0]$, and let us first assume $H$ is a tree.
It is easy to see that because in the original instance 
every node $v \in H$ can follow a directed path to some overloaded node in $H$, 
$v$ cannot receive all incident edges in $H$ without becoming overloaded. 
Suppose the lemma does not hold and consider a maximal tree $\overline{T} \subseteq H$ remaining outside 
of $\C^{\dagger}_0$. 
In the new instance, all edges of $H$ connecting $\overline{T}$ to the rest of the nodes in $H\backslash \overline{T}$ are directed toward $\overline{T}$. 
By induction, we can show that there is a node $v \in \overline{T}$ which receives all its incident edges in $H$, implying that 
$v \in \C^{\dagger}_0$, a contradiction. 

If $H$ is a cycle, observe that at least one node in $H$ must be overloaded in the new instance and hence part of $\C^{\dagger}_0$. Now we can proceed as before. 

\qed\end{proof}

\begin{lemma} In round 0, suppose that both instances are in the sub-procedure 
\emph{Conflict set construction} and $\C'_0 \subseteq \C^{\dagger}_0$.
In the new instance, as long as there is an edge 
$e=vu \not \in E(\C'_0)$ so that (1) it is directed toward $u$ in the original instance, 
(2) it is currently neutral in the new instance, and 
(3) $v \in \C^{\dagger}_0$ and $u \not \in \C^{\dagger}_0$, let $e$ be directed toward $u$ in the new instance. 
Then finally, an edge $e \not \in E(\C^{\dagger}_0)$, if directed in the original instance, is directed the same way 
in the new instance. 
\label{lem:finallyOutsideTheSame}
\end{lemma}

\begin{proof} Let $E_{\inc}(v)$ and $E_{\out}(v)$ denote the current set of incoming and outgoing edges of a node $v \not \in 
\C'_0$ in the original instance. In the new instance, after the fake orientations required in the lemma,
if every edge in $E_{\inc}(v)$ is directed toward $v$, then every edge in $E_{\out}(v)$ must be directed away from $v$, otherwise $v$ is overloaded. 

We now prove the lemma by contradiction. Suppose that edge $e_0 = v_0u \not \in E(\C^{\dagger}_0)$ is directed toward 
$u$ in the original instance while it is neutral or directed toward $v_0$ in the new instance after the fake orientations required 
in the lemma. In both cases $v_0 \not \in \C^{\dagger}_0$ (hence $v_0 \not \in \C'_0$) and $v_0$ is not overloaded. 
So there is an edge $e_1= v_1v_0 \in E_{\inc}(v_0)$ that is neutral or directed away from $v_0$ in the new instance after the fake orientations. 
As before, $v_1 \not \in \C^{\dagger}_0$. 
Repeating this argument, we conclude that the entire system is a cycle with no node in $\C^{\dagger}_0$ (hence also not in $\C'_0$), whose edges are all directed, say clockwise, in the original instance. Furthermore, in the new instance, each edge in the cycle is either neutral or directed counter-clockwise. Clearly, for at least one edge $xy$ in the cycle, 
it holds that $dl(x) + pl(x) + w_{xy} > (5/3 + \beta/3)t$. Since $x$ is not overloaded, this edge must have the same orientation (namely toward $y$) in both instances, a contradiction.
\qed\end{proof}

By Lemmas~\ref{lem:finallyTreesBigger} and~\ref{lem:finallyOutsideTheSame}, Lemma~\ref{lem:removeOriginal}(2)(3) hold, 
and we can apply Lemma~\ref{lem:imitationGame} to finish off all the remaining fake orientations in both instances while maintaining 
Lemma~\ref{lem:removeOriginal}(2)(3). 

The last thing to prove is that $\act'_0 \subseteq \act^{\dagger}_0$ after the activation rules are applied to both instances. 
If a node $v$ is overloaded in the original instance, by Lemma~\ref{lem:forcedOrientationLemma}, either its own pebble and dedicated load is already more than $(5/3 + \beta/3)t$, 
or it has a child $u \in \C'_0$ so that $pl(v)+dl(v)+ w_{vu} > (5/3 + \beta/3)t$. Thus, in the new instance, 
$v$ is either overloaded, or (as $u \in \C'_0 \subseteq \C^{\dagger}_0$) becomes a child of 
$u$ and is activated by Rule~1.
Furthermore, if a node $v$ is activated by Rule 1 in the original instance, then it has a father $u \in \C'_0$ satisfying $dl(v)+pl(v) + w_{vu} > (5/3+\beta/3)t$. As $u$ is also part of $\C^{\dagger}_0$ in the new instance, either $v$ is overloaded, or it is again activated by Rule 1. So we are sure Type A and Type B nodes of the original instance in $\act'_0$ must be part of $\act^{\dagger}_0$. Finally, as Lemma~\ref{lem:removeOriginal}(2) holds, nodes of $\act'_0$ activated by Rule 2 must also be 
part of $\act^{\dagger}_0$. This completes the proof of round 0 and 
the entire proof of Lemma~\ref{lem:removeOriginal}. 

\newpage

\bibliographystyle{splncs03}
\bibliography{bib_another}{}

\newpage
\appendix

\section{Improved Ratio for the 2-Valued Case} 

Suppose that $W\geq 2w$. 

As before, we first assume that $t < 2W$, and discuss the case $t\geq 2W$ at the end of the section. We modify our previous algorithm as follows:

\begin{definition} A node $v$ is 

\begin{itemize}
\item \emph{uncritical}, if $dl(v) + pl(v) \leq \thres -W - w$; 
\item \emph{critical}, if $dl(v) + pl(v) > \thres -W$;
\item \emph{hypercritical}, if $dl(v) + pl(v) > \thres$.
\end{itemize}
\end{definition}

\begin{quote} \textbf{Modified Algorithm~1}: As long as there is a bad system, apply \textsc{Explore1} and \textsc{Push} operation repeatedly. When there is no bad system left, return a solution with makespan at most $\thres$.
If at some point, \textsc{push} is no longer possible, declare that $\OPT \geq t+1$.
\end{quote}

The proof of Lemma~\ref{lem:simplePush} remains the same, and to establish Lemma~\ref{lem:firstContradiction} we just need to re-do the proof of Claim~\ref{cla:firstMain}.

\noindent \emph{New Proof of Claim~\ref{cla:firstMain}}: By the same reasoning as before, 
\begin{itemize}
\item none of the nodes in $\act(S)$ is uncritical;
\item if $S$ is a tree and $\act(S) \neq \emptyset$, at least one node $v\in \act(S)$ is critical; furthermore, if $|\act(S)| = 1$, this node $v$ satisfies $dl(v) + pl(v) > \thres-w$;
\item if $S$ is an isolated node $v \in \act$, then $dl(v) + pl(v) > \thres-w$. 
\end{itemize}

We now re-do the case analysis. 
\begin{enumerate}

\item Suppose that $S$ is a good system and $\act(S) \neq \emptyset$. Then either $S$ is a tree and $\act(S)$ contains exactly one critical (but not hypercritical) node, or $S$ is an isolated node, or $S$ is a cycle and has no critical node. In the first case, if $|\act(S)| \geq 2$, 
the LHS of (\ref{equ:2valAlmostTooMuch}) is at least
\begin{eqnarray*}
(\thres - W+1) + (|\act(S)|-1)(\thres -W-w+1) + (|\act(S)|-1)W = \\
|\act(S)|t - W + |\act(S)|(\lfloor \frac{W}{2} \rfloor +1) - (|\act(S)|-1) w > \\
|\act(S)|t + \frac{(|\act(S)|-2)W}{2} - (|\act(S)|-1) w \geq |\act(S)|t-w, 
\end{eqnarray*}
\noindent where the first inequality holds because $\lfloor \frac{W}{2} \rfloor +1 > \frac{W}{2}$ and the 
last inequality holds because $|\act(S)| \geq 2$ and $W \geq 2w$. If, on the other hand, $|\act(S)| = 1$, then the LHS of (\ref{equ:2valAlmostTooMuch}) is strictly more than 
\begin{eqnarray*}
\thres-w \geq t = |\act(S)|t, 
\end{eqnarray*}
and the same also holds for the case when $S$ is an isolated node.
Finally, in the third case, the LHS of (\ref{equ:2valAlmostTooMuch}) is at least 
\begin{eqnarray*}
|\act(S)|(\thres -W-w+1) + |\act(S)|W > |\act(S)|t.
\end{eqnarray*}

\item Suppose that $\act(S)$ contains at least two critical nodes, or that $S$ is a cycle and $\act(S)$ has 
at least one critical node. In both cases, $S$ is a bad system. Furthermore, the LHS of (\ref{equ:2valTooMuch}) 
can be lower-bounded by the same calculation as in the previous case with an extra term of $w$. 

\item Suppose that $\act(S)$ contains a hypercritical node. Then the system 
$S$ is bad, and the LHS of (\ref{equ:2valTooMuch}) is at least 

\begin{eqnarray*}
(\thres + 1) + (|\act(S)|-1)(\thres -W-w+1) + (|\act(S)|-1)W = \\ 
|\act(S)|(\thres+1) - (|\act(S)|-1)w > |\act(S)|t,
\end{eqnarray*}

\noindent where the last inequality holds because $W \geq 2w$. 
\qed
\end{enumerate}

\textbf{Approximation Ratio}: When $t \geq 2W$, we can again use the Gairing et al's algorithm~\cite{gairing04}, which either correctly reports that 
$\OPT \geq t+1$, or returns an assignment with makespan at most $t+W-1$. 

Suppose that $t$ is the smallest number for which an assignment is returned (then $\OPT \geq t$). Then the approximation ratio is 

$$ \frac{\thres}{\OPT}, \mbox{if $t < 2W$;\hspace*{0.3in}} \frac{t+W-1}{\OPT}, \mbox {if $t \geq 2W$}.$$ 

The former is bounded by $1 + \frac{\lfloor \frac{W}{2} \rfloor}{W}$, since $\OPT \geq W$; the latter 
is bounded by $1 + \frac{W-1}{2W} \leq 1 + \frac{\lfloor \frac{W}{2} \rfloor}{W} $, since 
$\OPT \geq t \geq 2W$. We can thus conclude: 

\begin{theorem} Suppose that $W \geq 2w$. With arbitrary dedicated loads on the machines, jobs of weight $W$ 
that can be assigned to two machines, and jobs of weight $w$ that can be assigned to any number of machines, 
we can find a $1 + \frac{\lfloor \frac{W}{2} \rfloor}{W}$ approximate solution in polynomial time.
\label{thm:firstTheoremAgain}
\end{theorem}

\end{document}